\documentclass[aps,pra,reprint,superscriptaddress,floatfix,amsmath,amssymb,amsfonts]{revtex4-1}
\pdfoutput=1
\usepackage[ascii]{inputenc}
\usepackage[bookmarksnumbered=true,hypertexnames=false,colorlinks=true,linkcolor=blue,urlcolor=blue,citecolor=blue,anchorcolor=green]{hyperref}
\usepackage[usenames,dvipsnames]{xcolor}
\usepackage{microtype,cleveref,graphicx,braket,amsthm,mathtools,bm,enumitem}
\newtheorem{thm}{Theorem}\crefname{thm}{Theorem}{Theorems}
\newtheorem{lem}{Lemma}\crefname{lem}{Lemma}{Lemmas}
\crefname{section}{Section}{Sections}\crefname{appendix}{Appendix}{Appendices}
\crefname{equation}{Eq.}{Eqs.}\crefname{figure}{Fig.}{Figs.}
\DeclareMathOperator{\sign}{sign}

\newcommand*{\Z}{\mathbb Z}
\newcommand*{\C}{\mathbb C}
\newcommand*{\R}{\mathbb R}

\newcommand*{\rmi}{\mathrm{i}}
\newcommand*{\rme}{\mathrm{e}}
\newcommand*{\rmd}{\mathrm{d}}
\newcommand*{\ot}{\otimes}
\newcommand*{\op}{\oplus}
\begin{document}

\title{Rigorous free fermion entanglement renormalization from wavelet theory}

\author{Jutho Haegeman}
\thanks{JH, BS, MW, VBS are primarily responsible for the content of this work.}
\affiliation{Department of Physics and Astronomy, Ghent University, Krijgslaan 281, S9, 9000 Gent, Belgium}

\author{Brian Swingle}
\thanks{JH, BS, MW, VBS are primarily responsible for the content of this work.}
\affiliation{Department of Physics, Massachusetts Institute of Technology, Cambridge, MA 02139, USA}
\affiliation{Department of Physics, Harvard University, Cambridge, MA 02138, USA}
\affiliation{Department of Physics, University of Maryland, College Park, MD 20742, USA}

\author{Michael Walter}
\thanks{JH, BS, MW, VBS are primarily responsible for the content of this work.}
\affiliation{Stanford Institute for Theoretical Physics, Stanford University, Stanford, CA 94305, USA}

\author{Jordan Cotler}
\thanks{JC collaborated on related approaches and assisted in preparing the manuscript.}
\affiliation{Stanford Institute for Theoretical Physics, Stanford University, Stanford, CA 94305, USA}

\author{Glen Evenbly}
\thanks{GE has independent unpublished data arriving at a similar 2D network structure and assisted in preparing the manuscript.}
\affiliation{D\'epartement de physique, Universit\'e de Sherbrooke, Sherbrooke, QC J1K 2X9, Canada}

\author{Volkher B.~Scholz}
\thanks{JH, BS, MW, VBS are primarily responsible for the content of this work.}
\affiliation{Department of Physics and Astronomy, Ghent University, Krijgslaan 281, S9, 9000 Gent, Belgium}
\affiliation{Institute for Theoretical Physics, ETH Z\"urich, Wolfgang-Pauli-Str. 27, 8093 Z\"urich, Switzerland}

\begin{abstract}
We construct entanglement renormalization schemes which provably approximate the ground states of non-interacting fermion nearest-neighbor hopping Hamiltonians on the one-dimensional discrete line and the two-dimensional square lattice. These schemes give hierarchical quantum circuits which build up the states from unentangled degrees of freedom. The circuits are based on pairs of discrete wavelet transforms which are approximately related by a ``half-shift": translation by half a unit cell. The presence of the Fermi surface in the two-dimensional model requires a special kind of circuit architecture to properly capture the entanglement in the ground state. We show how the error in the approximation can be controlled without ever performing a variational optimization.
\end{abstract}


\maketitle

\section{Introduction}
Characterizing quantum phases of matter and computing their low-temperature physical properties is one of the central challenges of quantum many-body physics. Part of the challenge is that many phases and phase transitions are best characterized in terms of their pattern of entanglement, as opposed to, say, their pattern of symmetry breaking~\cite{landau1937phasetransitions}. As an extreme case, topological phases of matter (see~\cite{wen2012toporeview} for a review) are solely characterized by entanglement, and since a topological state is gapped, it is insensitive to any local perturbation~\cite{hastings2005topostable,bravyi2010topostable} and so need not break any symmetry. By contrast, metallic states with Fermi surfaces have many low energy excitations, but they can also be usefully characterized in terms of their anomalously high entanglement. In this work we are concerned with such metallic states.

A useful idealization is to focus on ground state physics, where the entanglement entropy of spatial subsystems provides powerful non-local diagnostics of phases.
For example, most known ground states obey an ``area law": the entanglement entropy of a subsystem scales as the boundary size of the subsystem~\cite{eisert2010colloquium}.
Topological states have a subleading (in subsystem size) shape-independent ``topological entanglement entropy" term which partially characterizes the state~\cite{kitaev2006tee,levin2006tee}.
Another diagnostic is the scaling of entropy with subsystem size itself: metals logarithmically violate the area law when they have a Fermi surface~\cite{wolf2006violation,gioev2006entanglement,swingle2010entanglement}.

Yet, discussions based on entanglement entropy are only the beginning.
We can more fully characterize the pattern of entanglement in a quantum state by giving a recipe for physically constructing it from unentangled degrees of freedom.
Such a recipe is called a \emph{quantum circuit} and constitutes a sequence of local unitary transformations which produce the desired state from an unentangled state.
While symmetry breaking states, e.g.~a ferromagnet with all spins aligned, can be caricatured by an unentangled state (or more realistically by a state in which only nearby degrees of freedom are entangled), it is known that topological states must have a high degree of non-local entanglement as measured by the minimal number of circuit layers needed to prepare them from an unentangled state~\cite{bravyi2006topogeneration}.
The anomalously high entanglement in metallic states similarly implies that any circuit which prepares a metallic state starting from an unentangled state must have many layers.

In terms of calculations, quantum circuits often yield efficient classical algorithms for computing physical properties such as correlation functions. Moreover, they give rise to a local description of the multipartite entanglement structure in terms of multilinear maps. As such they form an important subclass of \emph{tensor network states}, a very successful variational class of quantum states which has been shown to be applicable in situations when other methods fail, e.g., due to a fermion sign problem in Monte Carlo methods (for a review see~\cite{murg2009tnsreview,orus2014tnsreview}).

In terms of experiments, quantum circuits provide a precise operational procedure, implementable on a sufficiently versatile quantum simulator or quantum computer, to prepare interesting states. For example, while it might be difficult to directly cool a system to its ground state, a quantum circuit provides an alternative way to directly engineer the ground state.

In this work we provide quantum circuits that, when acting on a suitable unentangled state, prepare approximations to the metallic ground states of certain one- and two-dimensional non-interacting fermion Hamiltonians.
This is a non-trivial task in part because these states of matter are highly entangled and violate the area law. We work primarily in the thermodynamic limit of an infinite lattice, but our constructions can also be applied to finite-size systems. The circuits themselves are composed of layers having a hierarchical renormalization group structure, in which the size of the system is doubled (or halved) after each layer.
In one dimension, the scheme is a version of the multiscale entanglement renormalization ansatz (MERA)~\citep{vidal2008mera}, while in two dimensions it yields an interesting branching network structure~\cite{evenbly2014branchmera}.
Such renormalization group inspired quantum circuits have been useful in describing a variety of gapless states in one dimension~\cite{montangero2009critmera,pfeifer2009critmera,corboz2010intfmera,evenbly2010ffmera}, and inroads have been made in two-dimensional systems~\cite{evenbly2010entanglement,aguado2008topomera,cincio20082dmera,evenbly20092dmera,swingle2016sqrt}.
They have also been instrumental for recent progress in our understanding of the holographic duality~\cite{swingle2012entanglement,qi2013exact,pastawski2015holographic,hayden2016holographic}.
As in the pioneering work~\cite{evenbly2016entanglement}, our construction is based on discrete wavelet transforms, although we take a somewhat different approach here.

The key features of our work are the following. Our circuits are hand-crafted---no variational optimization is used---and come with provable error bounds.
The essential physics is a certain ``half-shift" condition discussed in detail below~\cite{selesnick2002design}. Our two-dimensional circuit provides a representation of a state with a Fermi surface, albeit on a special lattice and at a special filling where the Fermi surface has no curvature.
The error scaling with circuit size is consistent with the hypothesis that an exact renormalization group circuit can be implemented using a quasi-local circuit with rapidly decaying tails~\cite{swingle2016ssourcery}.
Together, our results address a long standing challenge: to rigorously construct tensor networks for gapless states of matter with Fermi surfaces.

This paper is organized as follows.
We first first briefly set the stage for our work and review the basics of non-interacting fermions systems and some wavelet terminology (\cref{sec:setup}).
Then we discuss the one-dimensional model and construct MERA quantum circuits that approximate its ground state (\cref{sec:1d}).
Next, we generalize to the two-dimensional model which has a square Fermi surface (\cref{sec:2d}).
We then explain how to obtain \emph{a priori} error bounds for our constructions (\cref{sec:circuits and errors}), and we illustrate the effectiveness of our results through numerical experiments (\cref{sec:numerics}).
We conclude in \cref{sec:conclusions}.

\section{Setup}\label{sec:setup}
Throughout this paper we work in the context of non-interacting fermion systems.
At the single-particle level, the systems can be described by a Hilbert space spanned by a basis of single-particle states or \emph{modes}~$\phi_\alpha$.
This data depends on the physical setup, but in the cases below, $\phi_\alpha$ can be taken to be a single-particle state in which the particle is localized at site~$\alpha$ in some lattice~$\Z^d$.
The many-body description is achieved by \emph{second quantization}, i.e.,
passing to creation, $a_\alpha^\dagger$, and annihilation, $a_\alpha$, operators which create and destroy a fermion in the single-particle state~$\phi_\alpha$.

The Hamiltonians we study will all consist of fermion bilinears of the form~$H = \sum_{\alpha, \beta} a_\alpha^\dagger h^{(1)}_{\alpha,\beta} a_\beta$ where~$h^{(1)}_{\alpha,\beta}$ can be regarded as a single-particle Hamiltonian acting on the mode space.
The eigenstates of such a ``quadratic" Hamiltonian are many-body quantum states of fermions that obey Wick's theorem.
This in turn implies that these states are uniquely specified by the two point correlation function~$G_{\alpha,\beta} = \langle a_\alpha^\dagger a_\beta \rangle$.
In particular, the ground state of~$H$ is obtained from the state with no fermions by diagonalizing the matrix~$h^{(1)}_{\alpha,\beta}$ and placing one fermion in each mode with negative single-particle energy.
It is therefore possible to carry out much of the analysis of the ground state at the level of single-particle states.
In particular, the circuit approximation of the ground state is constructed from single-particle unitaries~$u = \mathrm{e}^{\mathrm{i} z}$.
The corresponding ``quadratic'' many-body unitary~$U = e^{\mathrm{i} Z}$ with~$Z = \sum_{\alpha,\beta} a_\alpha^\dagger z_{\alpha,\beta} a_\beta$ then acts as~$U^\dagger a_\alpha U = \sum_\beta u_{\alpha,\beta} a_\beta$.
Any state obeying Wick's theorem can always be prepared from an unentangled state (consisting of fermions localized to sites) by acting with such a ``quadratic'' unitary.

The models we study are translation invariant, so we immediately know how to diagonalize the single particle Hamiltonian $h^{(1)}$ using the Fourier transform (up to a small eigenvalue problem in case of having several modes per site).
However, the Fourier transform is \emph{not} a local unitary transformation (it mixes modes arbitrarily far away), so it fails to fully capture the special physics of the ground state and its real-space entanglement renormalization structure.
For example, the Fourier transform typically takes an unentangled state to a state with volume law entanglement.
Importantly, however, the ground state is invariant under basis transformations within the filled and empty spaces (negative and positive energy eigenspaces, respectively), where the former is also known as the \emph{Fermi sea}.
We can therefore prepare the same state by filling modes which are suitable linear combinations of negative energy eigenstates, chosen to be approximately local in real space.
Vice versa, we can \emph{approximate} the ground state by filling strictly local modes that are approximately supported within the Fermi sea.

Wavelets~\cite{mallat2008wavelet} provide a suitable basis to construct such local modes.
As first discussed in~\cite{evenbly2016entanglement}, the hierarchical structure of a wavelet transform provides the single-particle version of an entanglement renormalization quantum circuit.
In one dimension, wavelet transforms can be specified by a \emph{scaling filter}~$h_s$ and a \emph{wavelet filter}~$h_w$.
An input signal~$\psi\in\ell_2(\Z)$ (the space of square summable sequences) is then decomposed by convolution and downsampling into a scaling output~$\psi_s[n] = \sum_m h^*_s[m-2n] \psi[m]$ and a wavelet output~$\psi_w[n] = \sum_m h^*_w[m-2n] \psi[m]$.
Intuitively, the wavelet filter should project onto details of a certain scale, while the scaling filter should project on all features up to this scale.
The (discrete) wavelet transform is obtained by iterating this scheme: the scaling output is taken as the input signal for the next iteration.
It decomposes the Hilbert space into orthogonal subspaces, each describing details at a certain scale.
Its inverse reassembles the input signal from the wavelet outputs at all scales. 

If we design the wavelet transform such that it separates negative-energy from positive-energy modes, we obtain a renormalization scheme from the ``quadratic'' unitary~$U_{RG}$ corresponding to one step of the wavelet transform.
If the filters have finite length then~$U_{RG}$ is a finite-depth quantum circuit~\cite{evenbly2016representation}, meaning it is composed of a finite number of layers of two-site unitaries.
The unitary~$U_{RG}$ constitutes one renormalization step:
Given the ground state~$\ket\psi$ of the Hamiltonian,
\begin{equation*}
\ket\psi \approx  U_{RG} \bigl( \ket\psi \ot \ket\chi \bigr),
\end{equation*}
where on the right-hand side,~$\ket\psi$ is the ground state on the renormalized lattice and~$\ket\chi$ is some unentangled state.
Crucially, the disentangled sites are interleaved with the renormalized lattice and each unitary layer is a local transformation.
By composing many layers of~$U_{RG}$, we thus obtain a quantum circuit that approximately prepares the ground state.
The layout of the circuit is illustrated in \cref{fig:1dmera}.
The bottom of the figure corresponds to the state~$\ket\psi$, each layer of red and green blocks constitutes the quantum circuit implementing~$U_{\text{RG}}$, the product states~$\ket 1\ket 0$ on half of the sites make up~$\ket\chi$, and the lines which go up into the next layer correspond to~$\ket\psi$ on the other half of the sites, which can be identified with the renormalized lattice.
To realize this approach, we still need to design finite-length filters~$h_s,h_w$ such that the wavelet transform separates negative from positive energy modes.
We will now discuss in detail how this can be done systematically and to arbitrarily high fidelity for two fundamental model systems.

\section{Fermions on the discrete line}\label{sec:1d}
We first consider the fermion nearest-neighbor hopping Hamiltonian on the one-dimensional infinite discrete line,
\begin{equation}\label{eq:hamiltonian 1d}
H = - \sum_{n \in \Z} a_n^\dagger a_{n+1} + a_{n+1}^\dagger a_n.
\end{equation}
After blocking neighboring sites using the modes~$b_{1,n} = (-1)^n a_{2n}$ and~$b_{2,n} = (-1)^n a_{2n+1}$, corresponding to the even and odd sublattices, respectively, we can write
\begin{equation*}
H = -\sum_{n} b_{1,n}^\dagger b_{2,n} - b_{2,n}^\dagger b_{1,n+1} + b_{2,n}^\dagger b_{1,n} - b_{1,n+1}^\dagger b_{2,n}.
\end{equation*}
In terms of momentum modes~$b_j(k) = \sum_n b_{j,n} \rme^{-\rmi n k}$, the Hamiltonian is
\begin{equation}\label{eq:hamiltonian fourier}
   H = \int_{-\pi}^{\pi} \frac{\rmd k}{2 \pi} \begin{bmatrix} b_{1}(k)\\  b_{2}(k)\end{bmatrix}^\dagger\!\!\begin{bmatrix} 0 & \rme^{-\rmi k}-1\\ \rme^{\rmi k}-1 & 0 \end{bmatrix} \! \begin{bmatrix} b_{1}(k)\\  b_{2}(k)\end{bmatrix}.
\end{equation}
This is the discretized one-dimensional Dirac Hamiltonian using the staggered Kogut-Susskind prescription~\cite{kogut1975hamiltonian}.
The eigenmodes of the single-particle Hamiltonian~$h^{(1)}(k)$, i.e., the $k$-dependent $2 \times 2$~matrix in \cref{eq:hamiltonian fourier}, are
\begin{equation*}
  \phi_{\pm} (k) = \frac1{\sqrt 2} \begin{bmatrix} 1 \\ \pm \rmi \,\rme^{\rmi \frac{k}{2}}\end{bmatrix},
\end{equation*}
with energies~$\pm 2\sin(k/2)$ and velocities~$\pm \cos(k/2)$, corresponding to left ($-$) and right ($+$) movers~\footnote{Throughout this paper, we always consider the momenta~$k$ to be in~$(-\pi,\pi)$.}.
Thus, the many-body ground state is obtained by filling the negative energy eigenmodes~$\phi_{-\sign(k)}(k)$, corresponding to the Fermi sea~$[-\pi/2,\pi/2]$ in the original lattice.

\begin{figure}
  \includegraphics[width=0.45\textwidth]{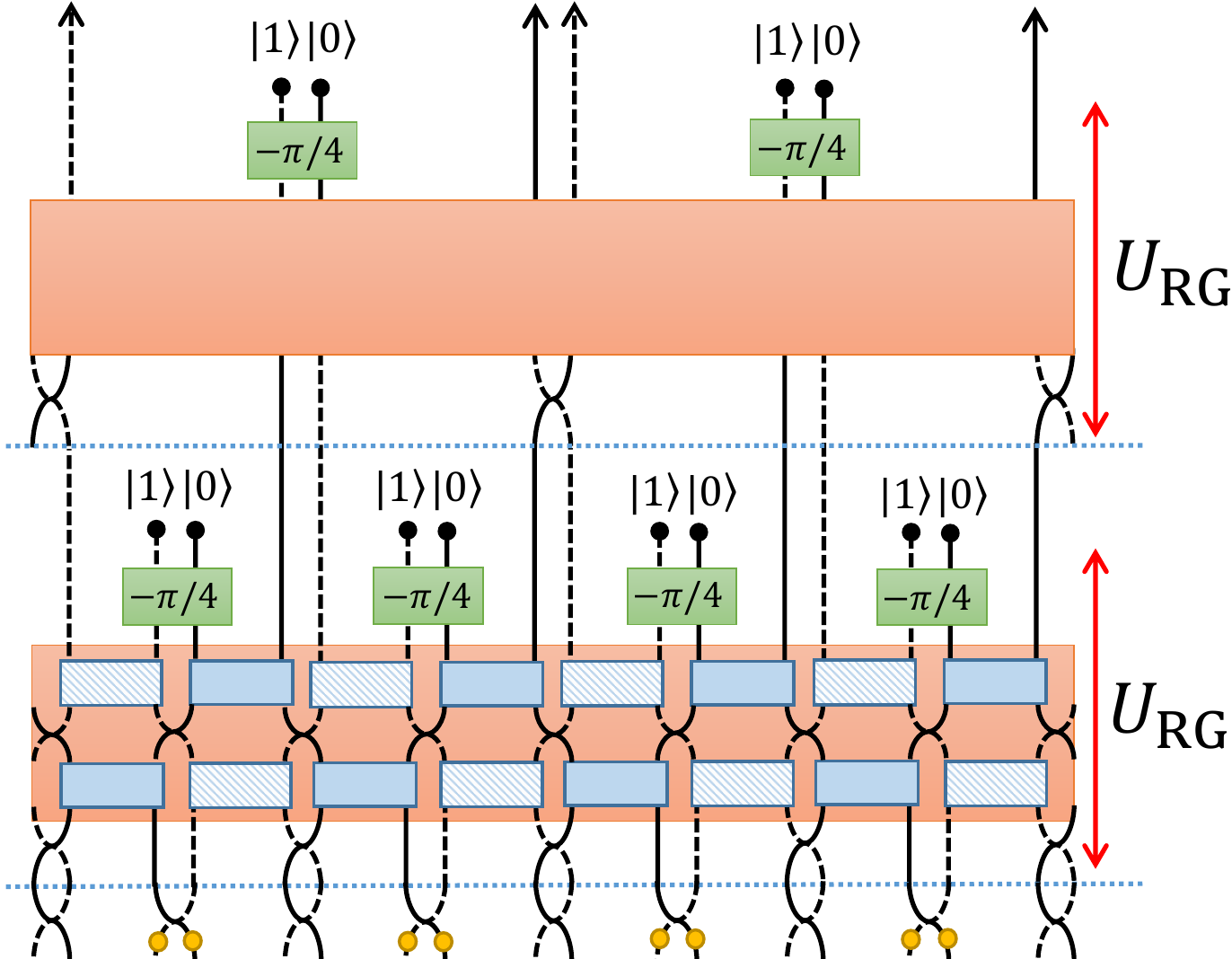}
  \caption{\emph{MERA quantum circuit preparing an approximate ground state of the one-dimensional hopping Hamiltonian~\eqref{eq:hamiltonian 1d},}
  using the notation of~\cite{evenbly2016entanglement}.
  To start, we apply phase gates (yellow circles) and swap sublattices in order to let the following gates act either on the even or the odd sublattice.
  The bulk of the MERA corresponds to two independent wavelet transforms, one for~$h_s$ (even sublattice) and one for~$g_s$ (odd sublattice).
  Each step of the two wavelet transforms gives rise to a layer (red boxes) of two parallel quantum circuits of depth~$K+L$, consisting of $2\times 2$ unitary gates (solid/hatched blue boxes).
  At the end of each layer, we apply second quantized Hadamard unitaries (green boxes) that couple both sublattices and we occupy the negative energy modes.
  The figure illustrates the case $K=L=1$ and $\mathcal L=2$ layers.
  As we increase~$K$, $L$, and the number of layers~$\mathcal L$, we systematically obtain better and better approximations to the ground state (see \cref{sec:circuits and errors}).}\label{fig:1dmera}
\end{figure}

To design a quantum circuit for the ground state, it is convenient to diagonalize the single-particle Hamiltonian into negative and positive energy eigenmodes by using the unitary~$u(k) = d(k) \, h_2$, where~$h_2 =\frac1{\sqrt2} \left[\begin{smallmatrix}1&~~1\\1&-1\end{smallmatrix}\right]$ is the Hadamard gate and~$d$ is of the form
\begin{equation}\label{eq:1dD}
d(k) \propto \begin{bmatrix} 1 & 0\\ 0 & -\rmi \sign(k)\rme^{\rmi k/2}\end{bmatrix},
\end{equation}
where, importantly, we are free to choose a~$k$-dependent phase.
Note that the matrix~$d(k)$ is discontinuous around~$k = 0$ because of the $\sign$~function, but not around~$k=\pm\pi$, where the discontinuity in the $\sign$ is cancelled by the discontinuity in the half-shift phase factor (and the result is even smooth).
The many-fermion ground state corresponding to the diagonalized single-particle Hamiltonian
\begin{equation}\label{eq:disentangled H}
u(k)^\dagger h^{(1)}(k) u(k) = \begin{bmatrix}-2\sin\left(\lvert k\rvert/2\right) & 0 \\ 0 & 2\sin\left(\lvert k\rvert/2\right)\end{bmatrix}.
\end{equation}
is disentangled and can be prepared in a completely local fashion by filling the even sublattice, corresponding to the first component, while leaving the odd sublattice empty.

We will now show that the ``quadratic'' unitary corresponding to~$u(k)$ can be well-approximated by a finite-depth quantum circuit.
The Hadamard~$h_2$ is not $k$-dependent and thus its second quantization simply corresponds to a local unitary between neighboring sites of the original non-blocked lattice.
Hence it suffices to focus on the unitary~$d(k)$, which is block-diagonal between the even and odd sublattice.
In view of the quantum circuit/wavelet correspondence discussed in \cref{sec:setup}, we thus need to design a \emph{pair} of wavelet transforms, acting on the even and odd sublattice and specified by filters~$h_s,h_w$ and~$g_s,g_w$, respectively, whose Fourier transforms are related by
\begin{equation}\label{eq:selesnickwavelet}
  g_w(k) \approx -\rmi \sign(k) \rme^{\rmi \frac k2} h_w(k).
\end{equation}
One can verify that \cref{eq:selesnickwavelet} is fulfilled if the corresponding scaling filters satisfy~\footnote{Here we use that we can choose the wavelet filters as the conjugate mirrors of the scaling filters. For concreteness, we will choose $h_w(k) = e^{ik} h_s^*(k+\pi)$ and likewise for~$g_w,g_s$. It is important to observe that~$k+\pi\equiv k-\sign(k)\pi\in(-\pi,+\pi)$ to deal correctly with the half-shift.}
\begin{equation}\label{eq:approxhalfdelay}
  h_s(k) \approx \rme^{\rmi \frac k2} g_s(k).
\end{equation}
The phase difference~$\rme^{\rmi k/2}$ in Fourier space implies that the two scaling filters are related by a \emph{half-shift} or half-delay in real space.
Its appearance is not surprising, given the translation invariance of the original (unblocked) lattice~\cite{kingsbury1999image}.
It is easily seen that the outputs of the inverse wavelet transforms are then at \emph{all} levels related approximately as in \cref{eq:1dD}, as illustrated in \cref{fig:wavelets1d}, and so can be used to implement~$d(k)$~\footnote{Since $\psi(n) = \sum_k h_s(n-2k) \psi_s(k) + h_w(n-2k) \psi_w(k)$, this is clear at the first level of the transform, but it can easily be shown to hold at all levels of the wavelet transform by using~\cref{eq:selesnickwavelet,eq:approxhalfdelay}; see \cref{app:proof}.}.
In other words, the same filters can be used throughout and a scale invariant circuit will be obtained.

\begin{figure}
\includegraphics[width=0.5\textwidth]{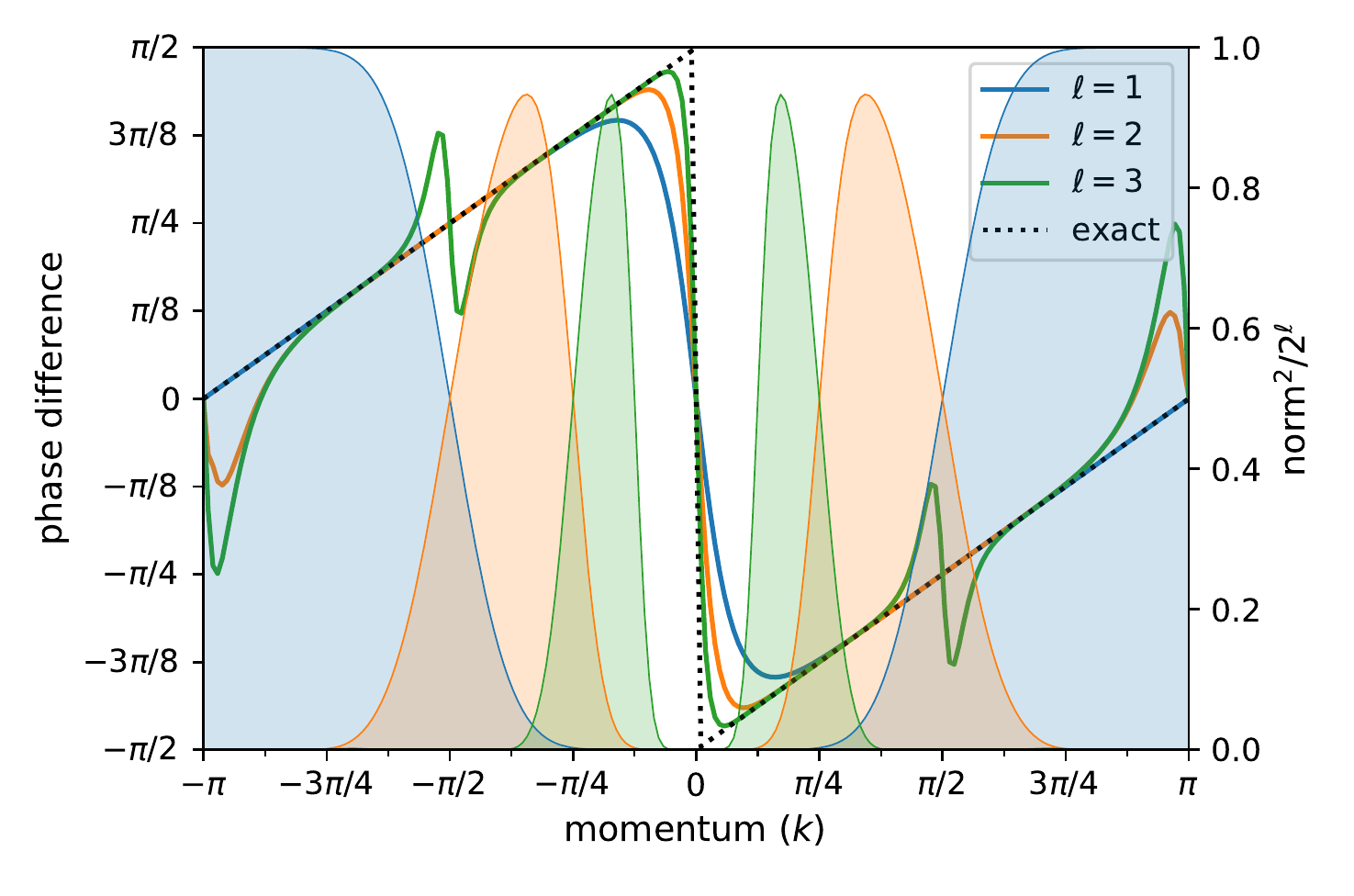}
\caption{\emph{Fourier spectrum of the outputs of the inverse wavelet transform} for levels $\ell=1,2,3$.
Both filters have equal magnitude in momentum space (shaded regions).
The relative phase difference of their Fourier transforms (solid lines) approximates the exact phase difference of the two components of~\cref{eq:1dD} or, equivalently, of the negative energy eigenstates of the single-particle Hamiltonian (black dotted line).}\label{fig:wavelets1d}
\end{figure}

Due to the discontinuity of the half-shift at~$k=\pm\pi$, a pair of local filters cannot satisfy~\eqref{eq:approxhalfdelay} exactly.
Fortunately, \emph{approximate} solutions were studied in great detail in the context of filter design in the signal processing literature~\cite{kingsbury1999image,selesnick2002design}.
Selesnick devised a general algorithm to construct filter pairs, indexed by two integers~$K$, $L$ and having length~$2(K+L)$, whose Fourier transforms have exactly equal magnitude and differ by a phase~$\rme^{\rmi \delta(k)}$~\cite{selesnick2002design}.
The parameter~$K$ determines the usual moment condition used in the wavelet literature, which controls the smoothness of the wavelets and the localization properties of the filters in momentum space.
The difference between~$\rme^{\rmi \delta(k)}$ and the ideal half-shift is controlled by the parameter $L$ and goes down quickly in the region around~$k=0$.
While~$\rme^{\rmi \delta(k)}$ is continuous at~$k=\pm \pi$ and therefore necessarily deviates from the half-shift in this region, the support of the scaling filter is, in this same region, suppressed with increased~$K$.
This allows us to control the error of the approximation~\eqref{eq:approxhalfdelay} by increasing the parameters~$K$ and~$L$ (see right panel of \cref{fig:renorm1d}).

We thus obtain entanglement renormalization quantum circuits by combing the circuits for the ``quadratic'' unitaries corresponding to the wavelet transforms, constructed using the procedure described in \cref{sec:setup}, with the Hadamard unitaries and the disentangled ground state of the diagonal Hamiltonian~\eqref{eq:disentangled H}.
These circuits, illustrated in \cref{fig:1dmera}, are composed of self-similar layers, each of which is a quantum circuit of finite depth $K+L$ that consists of nearest-neighbor $2\times 2$-unitary matrices.
This corresponds to a bond dimension~$\chi=2^{K+L}$ if the circuit is represented in the standard form of a binary MERA, written in terms of single layers of disentanglers and isometries~\cite{vidal2007entanglement,vidal2008mera}.
These quantum circuits allow us to rigorously approximate correlation functions of the ground state of the Hamiltonian~\eqref{eq:hamiltonian 1d} as discussed in \cref{sec:circuits and errors} and illustrated numerically in \cref{sec:numerics}.

From the perspective of the renormalization group, it is natural to consider the coarse-grained or \emph{renormalized Hamiltonian}.
Recall that the original single-particle Hamiltonian is of the form~$h^{(1)}(k) = e(k) (\cos(k/2) \sigma_y - \sin(k/2) \sigma_x)$, where~$e(k) = 2 \sin(k/2)$.
Because of the downsampling, both the wavelet and scaling outputs couple~$h^{(1)}(k)$ and~$h^{(1)}(k+\pi)$ (i.e., a single layer of a binary MERA is invariant under shifts over two sites).
The Hamiltonian can be naturally divided into three terms---corresponding to the scaling modes, the wavelet modes, and the mutual ``interaction'' between scaling and wavelet modes, respectively, each of which are a free fermion Hamiltonian.
The \emph{wavelet Hamiltonian} takes the exact single-particle form~$-\epsilon_\ell(k) \, \sigma_z$ (after the additional local Hadamard transforms).
Here, $\ell$ denotes the level of the wavelet transform, viz.\ the layer of the MERA, and~$\epsilon_\ell(k)>0$, so that its ground state is a product state in real space, obtained by filling the first mode on every site, in agreement with~\cref{eq:disentangled H}.
If \cref{eq:approxhalfdelay} is satisfied exactly, then the \emph{scaling Hamiltonian} or \emph{renormalized Hamiltonian} has the structure~$e_\ell(k) (\cos(k/2) \sigma_y - \sin(k/2) \sigma_x)$, where only the eigenvalues~$\pm e_\ell(k)$ change with the level~$\ell$, but not the eigenvectors; in general this is still true approximately.
This is the proclaimed scale invariance, and it provides an alternative way to see that the same pairs of scaling and wavelet filters should be used in every layer.
The coarse-grained dispersion relation~$e_\ell(k)$ does eventually reach a fixed point (up to a scaling $2^\ell$ that accounts for the rescaled lattice spacing), as illustrated in the left panel of \cref{fig:renorm1d}.
Note that there is also a residual \emph{wavelet-scaling interaction} term, originating from the overlap between the momentum space support of the wavelet and scaling filters, so that the Hamiltonian is not exactly block diagonal.
In particular, the dispersion relation~$e_{\ell+1}(k)$ is not simply~$e_\ell(k/2)$ (the lower half of the dispersion relation of the preceding level), and~$\lim_{\ell\to\infty} 2^\ell e_\ell(k)$ does not simply converge to~$\lvert k\rvert$ for all $k$ due to deviations around~$k=\pm \pi$ (see also the left panel in \cref{fig:renorm1d}).
An exact block diagonalization would require filters with non-overlapping support, which are therefore nonlocal in real space. While this behavior is more closely approximated with increasing~$K$, the magnitude of the interaction term decays at most polynomially in~$K$.
However, full block diagonalization of $H$ is too strong of a requirement and would allow the creation of arbitrary eigenstates by replacing the product states with a plane-wave state within a single layer.
For the ground state itself, convergence of correlation functions is still exponential in~$K$ and~$L$, as discussed in \cref{sec:circuits and errors}.

\begin{figure}
\includegraphics[width=0.5\textwidth]{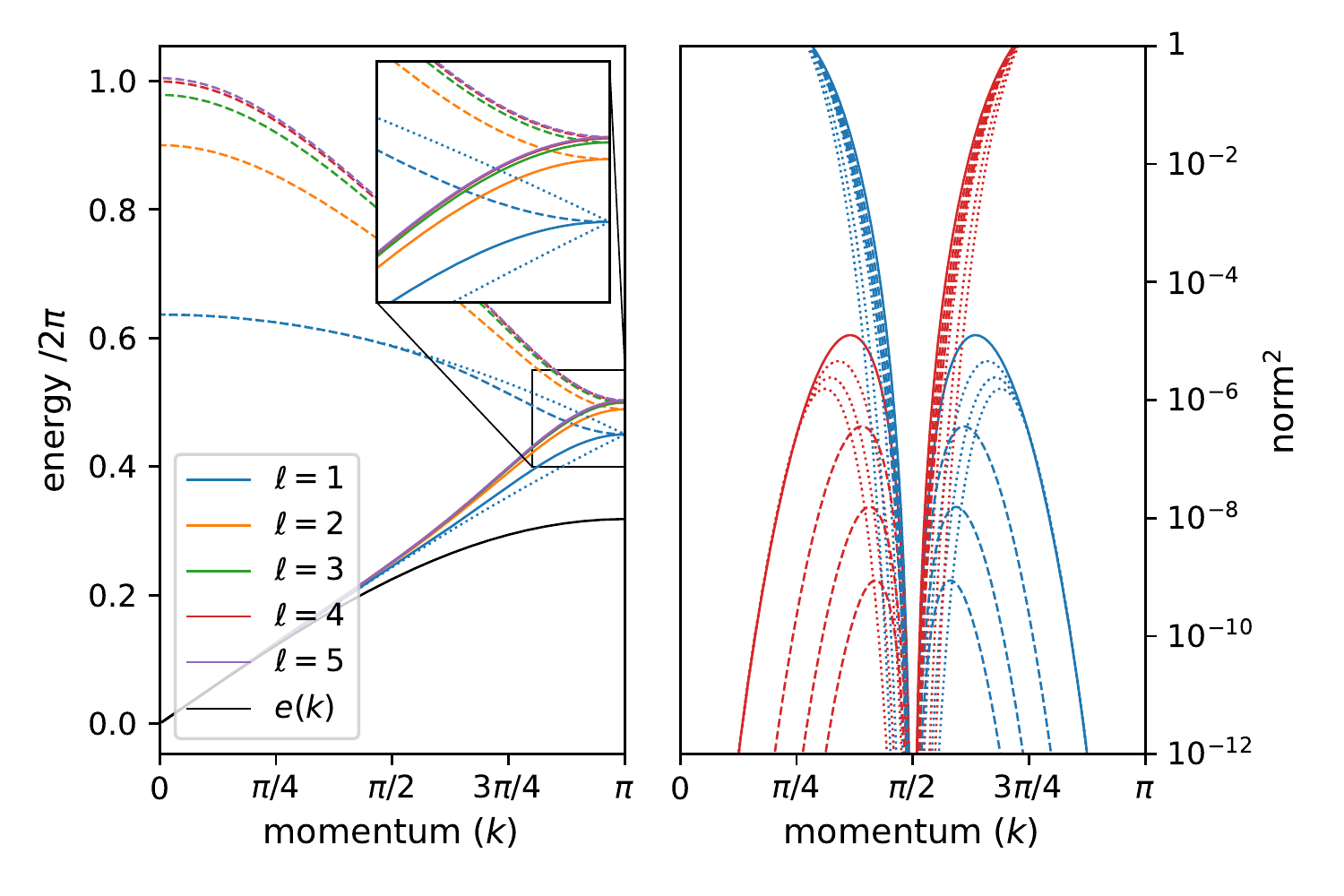}
\caption{\emph{Left panel:} Plots of the eigenvalues~$2^\ell e_\ell(k)$ of the scaling Hamiltonian (solid lines) and the eigenvalues~$2^\ell \epsilon_\ell(k)$ of the wavelet Hamiltonian (dashed lines) for levels~$\ell=1,\ldots,5$.
We observe convergence to a fixed-point Hamiltonian.
The original dispersion~$e(k)$ is shown by the solid black line.
The residual off-diagonal interaction between scaling and wavelet modes explains why the eigenvalues~$e_1(k),\epsilon_1(k)$ deviate slightly from~$e(k/2),e(k/2+\pi)$ (dotted blue lines).\\
\emph{Right panel:}
The single-particle modes obtained from level $\ell=1$ of the wavelet transform, translated back to the original lattice before blocking, should have momentum space support inside (blue) or outside (red) the Fermi sea~$[-\pi/2,+\pi/2]$.
While the wavelets exactly vanish at the Fermi points, small errors originate from the side lobes on the opposite side of the Fermi points.
The solid lines are~$K=L=4$.
For fixed~$L=4$ and $K=6,8,10$, the support is pushed away from the Fermi surface but the magnitude of the side lobes does not decrease much (dotted lines).
For fixed~$K=4$ and $L=6,8,10$, the side lobes appear to decrease exponentially fast (dashed lines).}\label{fig:renorm1d}
\end{figure}

\section{Square lattice and Fermi surface}\label{sec:2d}

We now extend our construction to fermions hopping on the two-dimensional infinite square lattice:
\begin{equation}\label{eq:hamiltonian 2d}
  H = - \sum_{m,n\in\Z} a_{m,n}^\dagger a_{m+1,n} + a_{m,n}^\dagger a_{m,n+1} + h.c.
\end{equation}
We again start by focusing on the single-particle domain, and then later transform everything into second-quantized form.
The two-dimensional problem we study is special because of the Fermi surface structure: the two-dimensional fermion Greens function $\langle a^\dagger_{x,y} a_{0,0} \rangle$ factorizes into two one-dimensional Greens functions, one which depends $x+y$ and one which depends on $x-y$. Thus, as in the one-dimensional case, we decompose the lattice into an even and odd sublattice, now defined by demanding that the sum of both coordinates is even or odd, respectively; and we likewise shift the Brillouin zone by momentum~$(\pi,\pi)$, resulting in new mode operators~$b_{1,x,y}=(-1)^{x+y} a_{x+y,x-y}$ and~$b_{2,x,y}=(-1)^{x+y} a_{x+y+1,x-y}$, with corresponding momentum modes~$b_i(k_x,k_y)$, $i=1,2$.
Note that these momenta are now defined with respect to the even/odd sub-lattice and hence are rotated by~$45$ degrees with respect to the original lattice.
This transformation effectively decouples the~$x$ and~$y$ direction, as the corresponding one-particle Hamiltonian is now of the form
\begin{equation*}
  h^{(1)} = -\begin{bmatrix} 0 & (1 - \rme^{-\rmi k_x})(1-\rme^{-\rmi k_y})\\ (1 - \rme^{+\rmi k_x})(1-\rme^{+\rmi k_y})&0\end{bmatrix}.
\end{equation*}
Its eigenvalues are products of the eigenvalues in the one-dimensional case, $\pm 4 \sin(k_x/2) \sin(k_y/2)$, with eigenmodes
\[ \phi_\pm(k_x,k_y)=\frac1{\sqrt2}\begin{bmatrix}1\\\pm\rme^{\rmi\frac{k_x+k_y}2}\end{bmatrix}. \]
As in the one-dimensional case, the eigenmodes exhibit a phase difference between the two sub-lattices corresponding to a half-shift in real space, but now the half-shift is in both lattice directions.
The positive and negative energy eigenmodes are given by $\phi_{\pm\sign(k_x)\sign(k_y)}(k_x,k_y)$, respectively, and are thus discontinuous around both~$k_x=0$ and $k_y=0$, as illustrated in the left panel of \cref{fig:2dhopping}.

It is now clear that we can diagonalize the single-particle Hamiltonian with the unitary $u(k) = d(k_x) d(k_y) \, h_2$,
where~$d$ is the block-diagonal unitary~\eqref{eq:1dD} and $h_2$ the Hadamard gate defined previously.
We can implement
\begin{equation}\label{eq:2dD}
d(k_x) d(k_y) \propto \begin{bmatrix} 1 & 0\\ 0 & - \sign(k_x)\sign(k_y)\rme^{\rmi (k_x+k_y)/2}\end{bmatrix}
\end{equation}
using the tensor products of two one-dimensional wavelet transforms as before---one acting in the $x$-direction and the other in the $y$-direction.
More specifically, let us denote by~$W \psi = \psi_s \op \psi_w$ a single step of the one-dimensional wavelet transform with filters~$h_w,h_s$.
Then
\begin{equation}\label{eq:outputWW}
  (W \ot W) \psi = \psi_{ss} \op \psi_{sw} \op \psi_{ws} \op \psi_{ww},
\end{equation}
which we identify as a single step of the two-dimensional separable wavelet transform.
In particular, the wavelet-wavelet component~$\psi_{ww}$ corresponds to the filter $h_{ww}(k_x,k_y)=h_w(k_x)h_w(k_y)$, and similarly if we use the one-di\-men\-si\-o\-nal filters $g_s,g_w$ instead.
Thus, \cref{eq:selesnickwavelet} implies that
\begin{equation*}
  h_{ww}(k_x,k_y) \approx -\sign(k_x)\sign(k_y) \rme^{\rmi \frac{k_x+k_y}{2}} \,g_{ww}(k_x,k_y),
\end{equation*}
which is precisely the desired phase relation between the two components of $d(k_x) d(k_y)$ (see left panel in \cref{fig:2dhopping}).
To obtain the tensor product of the two wavelet transforms, we now iteratively apply~$W\ot W$ to the scaling-scaling component~$\psi_{ss}$, as well as $W \ot I$ to~$\psi_{sw}$ and $I \ot W$ to~$\psi_{ws}$~\footnote{In contrast, in the two-dimensional separable wavelet transform the three components $\psi_{sw},\psi_{ws},\psi_{ws}$ together usually make up the detail coefficients.}.
The resulting transform is thus labeled by two levels, $\ell_x$ and $\ell_y$, corresponding to the number of renormalization steps in each direction.

\begin{figure}
\includegraphics[width=0.23\textwidth]{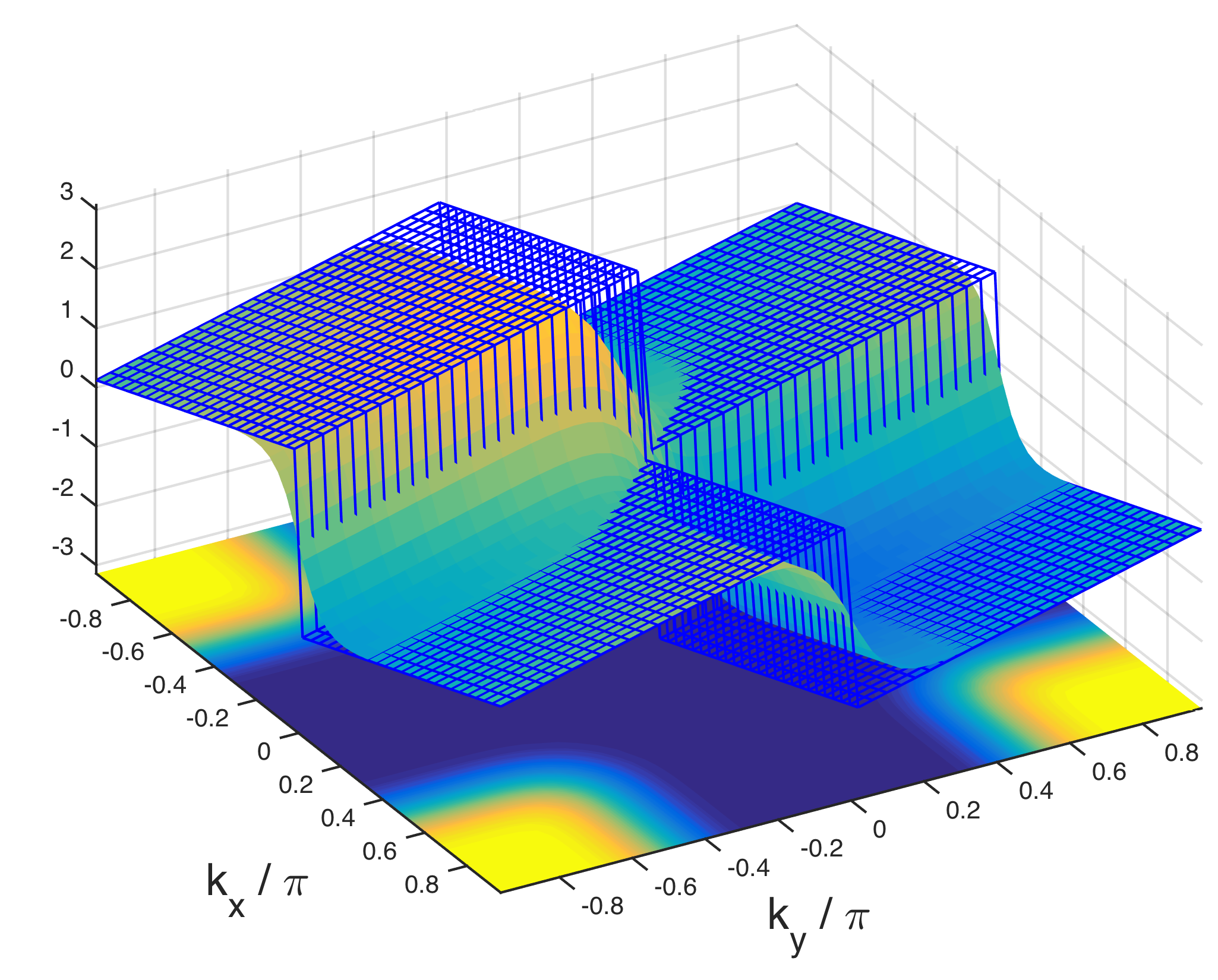}\includegraphics[width=0.23\textwidth]{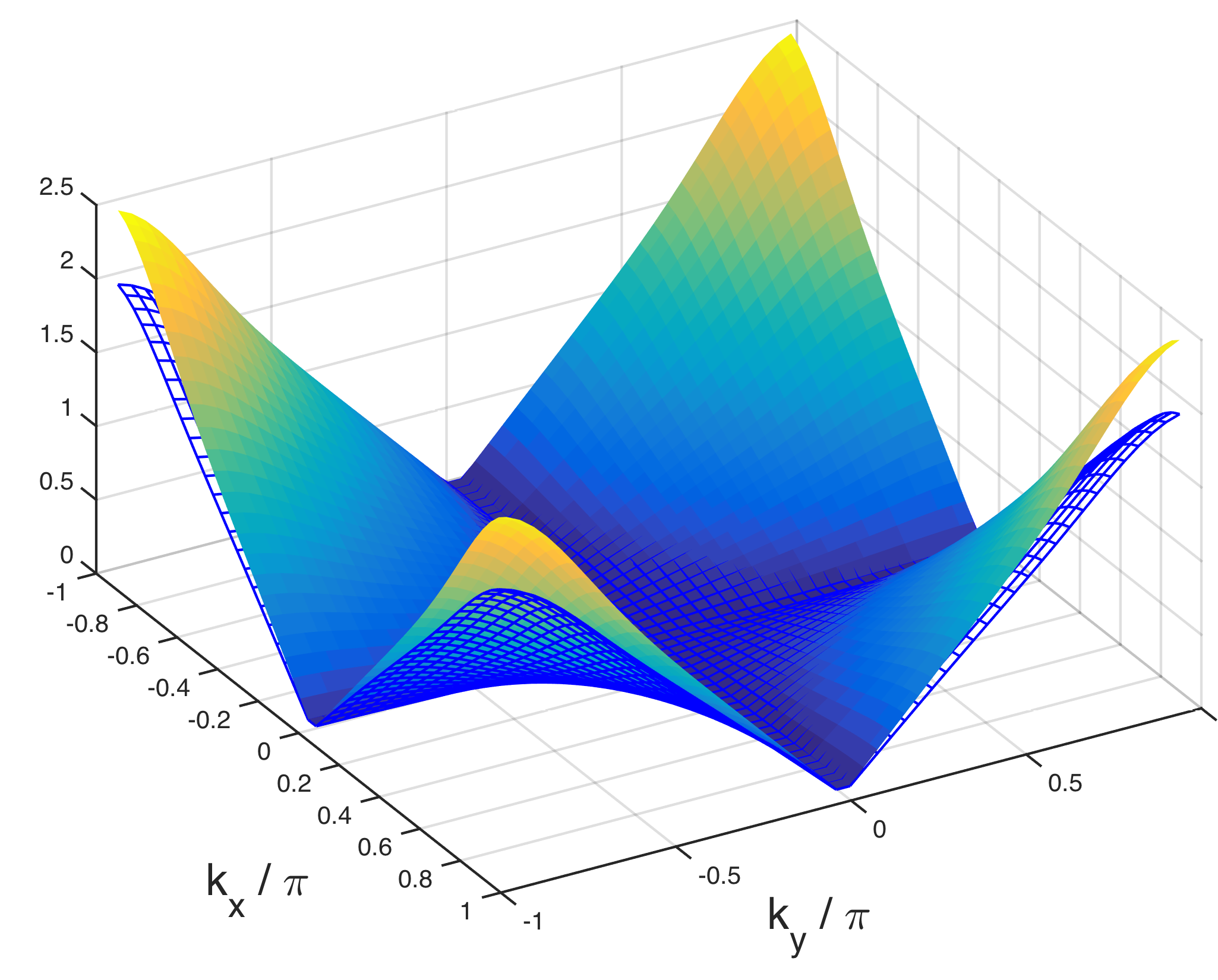}
\caption{\emph{Left panel:}
The relative phase difference between the Fourier transforms of~$h_{ww},g_{ww}$ (smooth surface) approximates the exact phase difference of the two components of~\cref{eq:2dD} or, equivalently, of the negative energy eigenstates of the single-particle Hamiltonian (wireframe mesh).
The colored shading of the coordinate plane indicates the momentum space support of $h_{ww}$ and $g_{ww}$.
It is concentrated around $k_x = k_y = \pm \pi$ and vanishes for $k_x=0$ or $k_y=0$. \\
\emph{Right panel:} The positive energy branch of the original Hamiltonian (wireframe mesh) and of the renormalized Hamiltonian (smooth surface) after 6 layers, where it has approximately reached its fixed point. The eigenmodes of both Hamiltonians are exactly characterized by the relative phase difference displayed in the left panel.}
\label{fig:2dhopping}
\end{figure}

\begin{figure}
\includegraphics[width=0.45\textwidth]{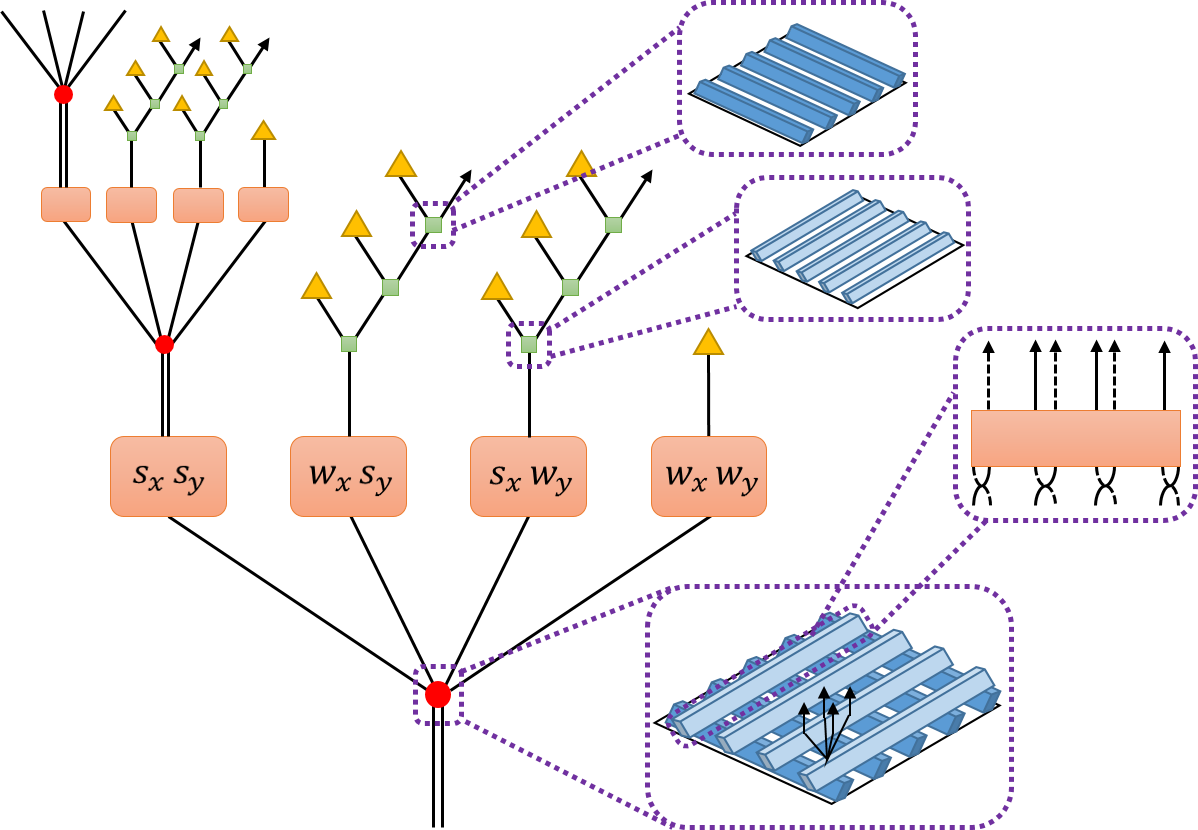}
\caption{\emph{Branching MERA quantum circuit preparing an approximate ground state of the two-dimensional hopping Hamiltonian~\eqref{eq:hamiltonian 2d}.}
The bottom layer consists of applying one-di\-men\-si\-o\-nal wavelet transforms in both the $x$ and $y$ direction of the $45$ degree rotated lattice.
This gives rise to four outputs, as illustrated in detail in \cref{fig:lattice cg}.
The wavelet-wavelet output corresponds to approximate eigenstates, and hence can be projected out after a subsequent Hadamard transform that couples the even and the odd sublattice, indicated by yellow triangles.
The mixed wavelet-scaling outputs are disentangled in one direction and therefore have to be processed further by one-dimensional transformations in the scaling directions, giving rise to two branches that take the form of binary trees as in \cref{fig:1dmera}.
The scaling-scaling output contains no disentangled degrees of freedom and is thus connected to a copy of the circuit on the renormalized lattice.}\label{fig:mera2d}
\end{figure}

After second quantization and converting these transformations into a quantum circuit, we obtain an entanglement renormalization quantum circuit of the form sketched in \cref{fig:mera2d,fig:lattice cg}.
This is a particular example of a \emph{branching MERA}, which generalizes the MERA to allow for logarithmic corrections to the area law~\cite{evenbly2014scaling,evenbly2014real,evenbly2014branchmera} and which was explicitly built with Fermi surfaces in mind.
Unlike in the original proposal, our network has \emph{three} branches instead of two.
Indeed, after each layer we are left with four branches, of which only one can be projected into a product state while the remaining three have to be analyzed further.
However, only one of the three branches keeps on branching in the higher levels. The other two are further disentangled by ordinary one-dimensional MERAs as in \cref{fig:1dmera}.
This ensures that the ground state produced by our network satisfies an appropriate area law of the form $S(R) \simeq R \log_2 R$ for the entropy of the reduced density matrix of an $R \times R$ box.
Indeed, let us first recall the estimation of the entanglement entropy in a one-dimensional MERA.
Each layer is a finite-depth quantum circuit that increases the entanglement entropy of a region by at most a constant amount $c_1$, so we obtain $S_{\text{1D}}(R) \leq c_1 + S_{\text{1D}}(R/2) \leq \ldots \leq c_1 \log_2 R$.
For a regular two-dimensional MERA, every layer can increase the entanglement entropy of an $R \times R$ box by $c_2 R$, leading to $S_{\text{2D}}(R) \leq c_2 R + S_{\text{2D}}(R/2) \leq \ldots \leq 2c_2 R$.
Thus, the entanglement entropy in a regular two-dimensional MERA obeys a strict area law.
In contrast, our branching MERA adds in every layer the entanglement contribution of a collection of one-dimensional MERAs in the horizontal and vertical direction. The resulting entanglement entropy is bounded by
\begin{align*}
S(R) &\leq c_2 R + 2 (R/2) S_{\text{1D}}(R/2) + S(R/2)\\
&\leq c_2 R + c_1 R \log_2(R/2) + S(R/2)
\leq \ldots \\
&\lesssim 2 c_1 R \log_2 R + (2 c_2 - 4 c_1) R.
\end{align*}
While only an upper bound, this estimate illustrates how a logarithmic violation of the area law can be obtained due to the one-dimensional MERAs in each layer.

\begin{figure}
\centering\includegraphics[width=8.5cm]{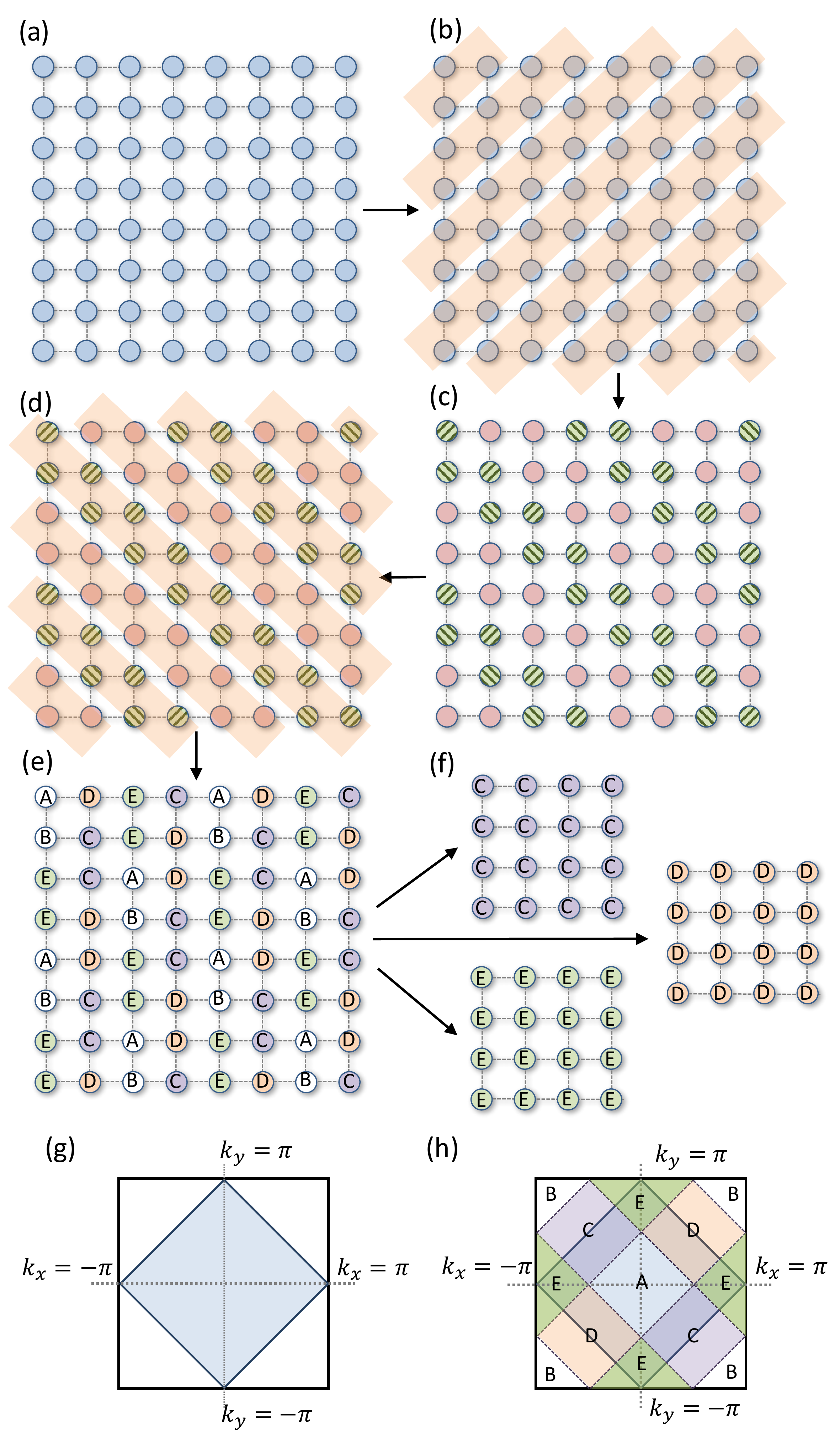}
\caption{\emph{One layer of the branching MERA for two-dimensional free fermions.}
(a) The 2D lattice of fermionic modes.
(b) A series of 1D wavelets transformations, as depicted in Fig. 1, are applied diagonally.
(c) Lattice sites correspond either to scaling (solid) or wavelet (striped) outputs from the preceding transformations.
(d) A series of 1D wavelets transformations are applied across the other diagonal.
(e) The lattice now contains five different types of site (labelled A--E).
(f) Sites labelled A, B are differentiated by the application of Hadamards and are truncated, while new sublattices are separately formed from sites C--E.
(g) The Brillouin zone of the 2D free fermions, where the shaded region denotes the Fermi sea.
(h) The sites in~(e) approximately correspond to the distinct regions of the Brillouin zone as shown.
Sites A and B contain modes in the occupied $\ket1$ or unoccupied $\ket0$ states respectively and may be truncated.
Sublattices from sites C and D consist of products of $1D$ chains that are only correlated in the back-sloping diagonal or forward-sloping diagonal directions respectively.
The sublattice from sites E is self similar to the original lattice.}
\label{fig:lattice cg}
\end{figure}

From the perspective of the renormalization group, the scaling-scaling branch gives rise to a renormalized Hamiltonian whose eigenmode structure is exactly the same as that of the original Hamiltonian, so that it can indeed be further processed in a self-similar fashion.
The eigenvalues of the renormalized Hamiltonian converge to a fixed point upon successive coarse-graining (see right panel of \cref{fig:2dhopping}).
The other two branches, resulting from a scaling filter in one direction and a wavelet filter in the other direction, give rise to coarse grained Hamiltonians depicted in \cref{fig:2dham1d}.
The structure of their eigenmodes is purely one-dimensional.
Indeed, for both outputs, one direction is already of wavelet type, so we only have to apply the one-dimensional MERA in the other direction to obtain wavelet outputs at each scale.

\begin{figure}
\includegraphics[width=0.23\textwidth]{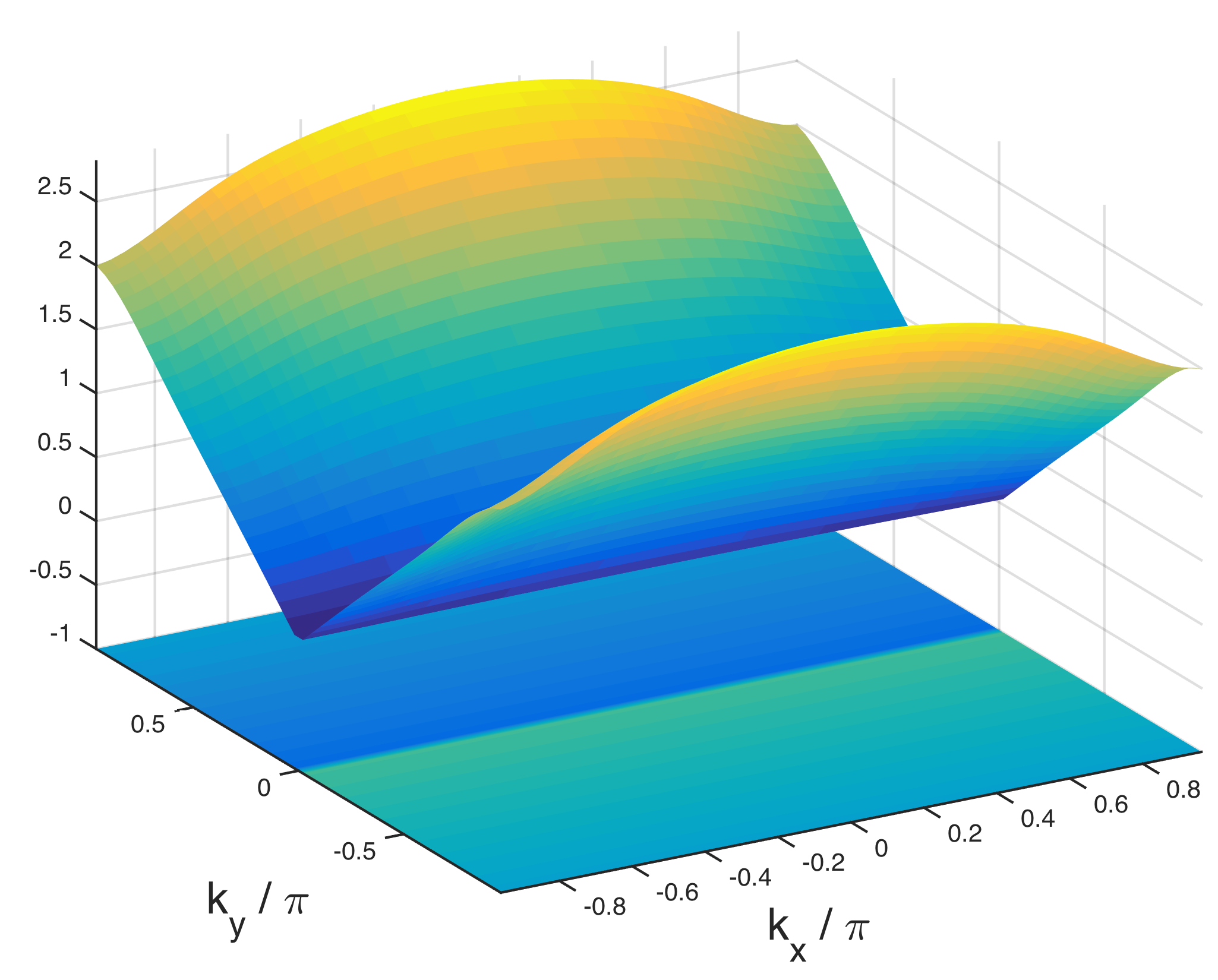}\includegraphics[width=0.23\textwidth]{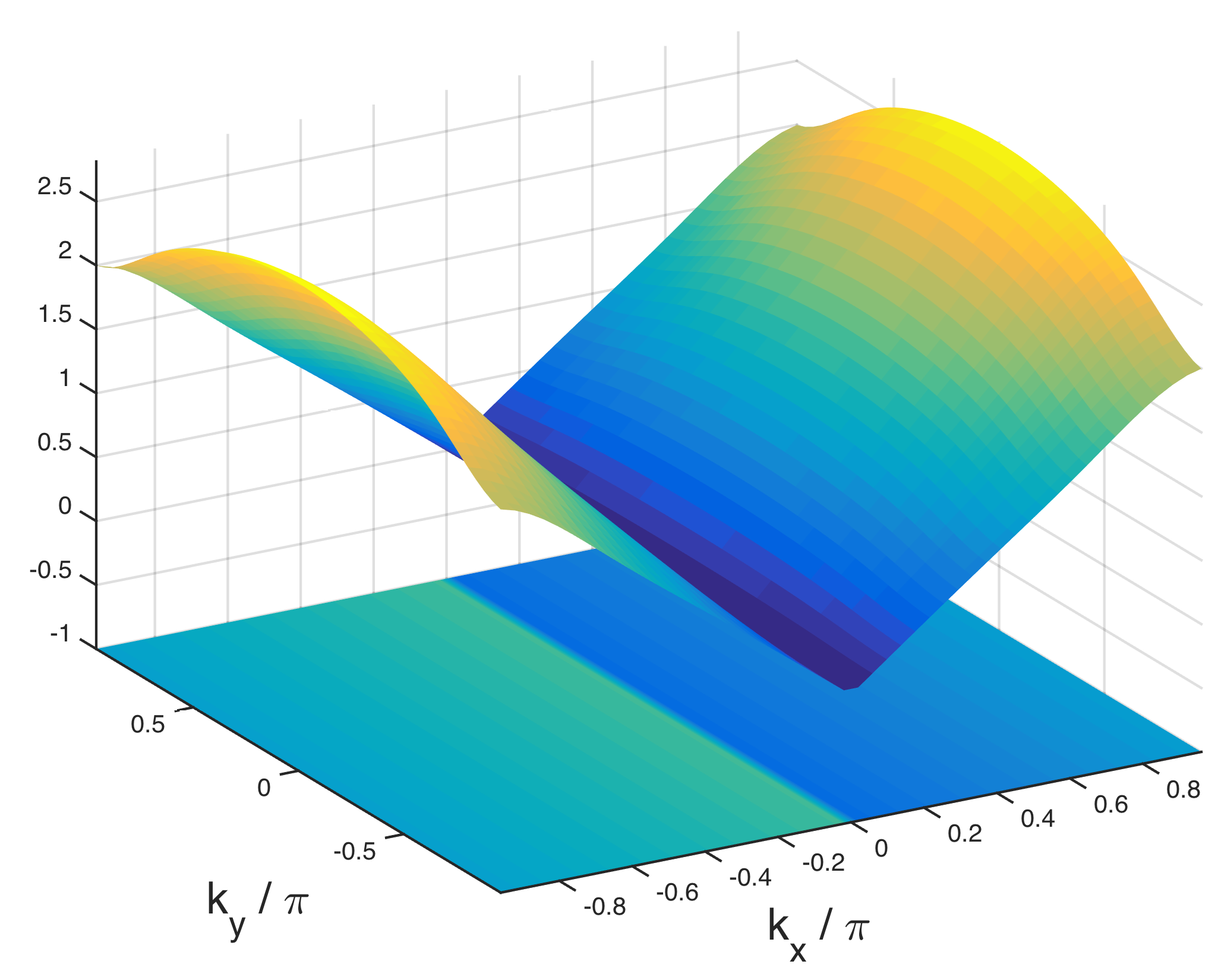}
\caption{\emph{Coarse-grained Hamiltonians for the wavelet-scaling and scaling-wavelet branches.}
We show the positive energy branch (surface) and the relative phase difference between the two components of the positive energy eigenmodes (color coding of the coordinate plane).
Because the eigenmodes are independent of one of the two directions and take the form of the one-dimensional problem studied in \cref{sec:1d}, the ground state can be disentangled by a tensor product of one-dimensional MERAs as in \cref{fig:1dmera}.}\label{fig:2dham1d}
\end{figure}

\section{Rigorous error estimates}\label{sec:circuits and errors}
In \cref{sec:1d,sec:2d} we constructed entanglement renormalization quantum circuits to approximately prepare the ground state of free fermion Hamiltonians in one and two dimensions, and we gave a heuristic account of the improved quality of our approximations with increased circuit parameters $K$ and $L$.
We will now discuss how this intuition can be turned into rigorous \emph{a priori} error estimates.
For simplicity, we will only formulate our result in the one-dimensional case, but its statement and proof are completely analogous for two dimensions.

Our theorem is stated in terms of correlation functions of fermion creation and annihilation operators.
Given a sequence~$f \in \ell_2(\Z)$, we define the corresponding annihilation and creation operators via $a(f) = \sum_{n\in\Z} f[n] a_n$ and $a^\dagger(f) = \sum_{n\in\Z} f[n]^* a^\dagger_n$.
We are interested in computing \emph{correlation functions} of $2N$ creation and annihilation operators in a many-body state $\Psi$,
\[
G(\{f_i\})_\Psi = \braket{\Psi | a^\dagger(f_1) \ldots a^\dagger(f_N) a(f_{N+1}) \ldots a(f_{2N}) | \Psi}.
\]
Other orderings of operators can be obtained by using the canonical anticommutation relations $\{ a(f), a(g)^\dagger \}=\braket{f|g}$ and $\{a(f),a(g)\}=0$.
The number of creation and annihilation operators must be equal to obtain a nonvanishing result since we are interested in states that are invariant (up to an overall phase) under a global $U(1)$ (particle number) transformation of the form $a_\alpha \mapsto e^{\mathrm{i} \theta} a_\alpha$. For a pure state of a finite size system, this invariance would simply imply that the state has a fixed number of particles. Let $D(\{f_i\})$ denote the maximal support of any linear combination of the observables $f_i$ (e.g., $n$ for an $n$-point function).
We will find that correlation functions of sparse observables are easier to approximate.

Our result is independent of any specific filter construction and only depends on the following parameters.
Let $h_s$ and $g_s$ be two scaling filters of finite length $M$ such that the half delay condition~\eqref{eq:approxhalfdelay} is approximately satisfied:
\begin{equation}\label{eq:approxhalfdelay quantitative}
  \lvert h_s(k) - \rme^{\rmi \frac k2} g_s(k) \rvert \leq \varepsilon < 1.
\end{equation}
We also assume that the filters generate corresponding multiresolution analyses with scaling functions bounded in absolute value by some constant $B\geq1$.
Then we have the following \emph{a priori} error estimate:

\begin{thm}\label{thm:main}
Let $\ket\Omega$ denote the exact ground state of the Hamiltonian~\eqref{eq:hamiltonian 1d} and $\ket{\Omega_\text{MERA}}$ the many-body state prepared by $\mathcal L$ layers of the MERA quantum circuit constructed from two scaling filters as above.
Then we have the following error bound for correlation functions:
For all $f_1,\dots,f_{2N}$ with $\lVert f_i\rVert_2 \leq 1$,
\begin{gather*}
  \bigl\lvert G(\{f_i\})_\Omega - G(\{f_i\})_{\Omega_{\text{MERA}}}\bigr\rvert \\
  \leq 24 \sqrt{N} \sqrt{C 2^{-\mathcal L/2} + 6 \varepsilon \log_2^2(C/\varepsilon)}
\end{gather*}
where the constant $C$ is given by $2^{3/2} \sqrt{D(\{f_i\})} \, B M$.
\end{thm}

\Cref{thm:main} shows that correlation functions can be approximated to arbitrarily high fidelity for a MERA constructed from suitable scaling filters.
As discussed in \cref{sec:1d}, Selesnick's algorithm gives rise to such filters, parametrized by two integers $K$ and $L$~\cite{selesnick2002design}.
Their length is $M=K+L$, and we numerically find that $B$ remains bounded, while $\varepsilon$ decreases exponentially as we increase $K$ and $L$ (see \cref{fig:selesnick filter error}).
Thus the error bound in \cref{thm:main} is likewise exponentially small if the number of layers $\mathcal L$ is sufficiently large.
We illustrate this in \cref{sec:numerics} below, where we numerically approximate the energy density and more general two-point functions.

\begin{figure}
  \includegraphics[width=0.5\textwidth]{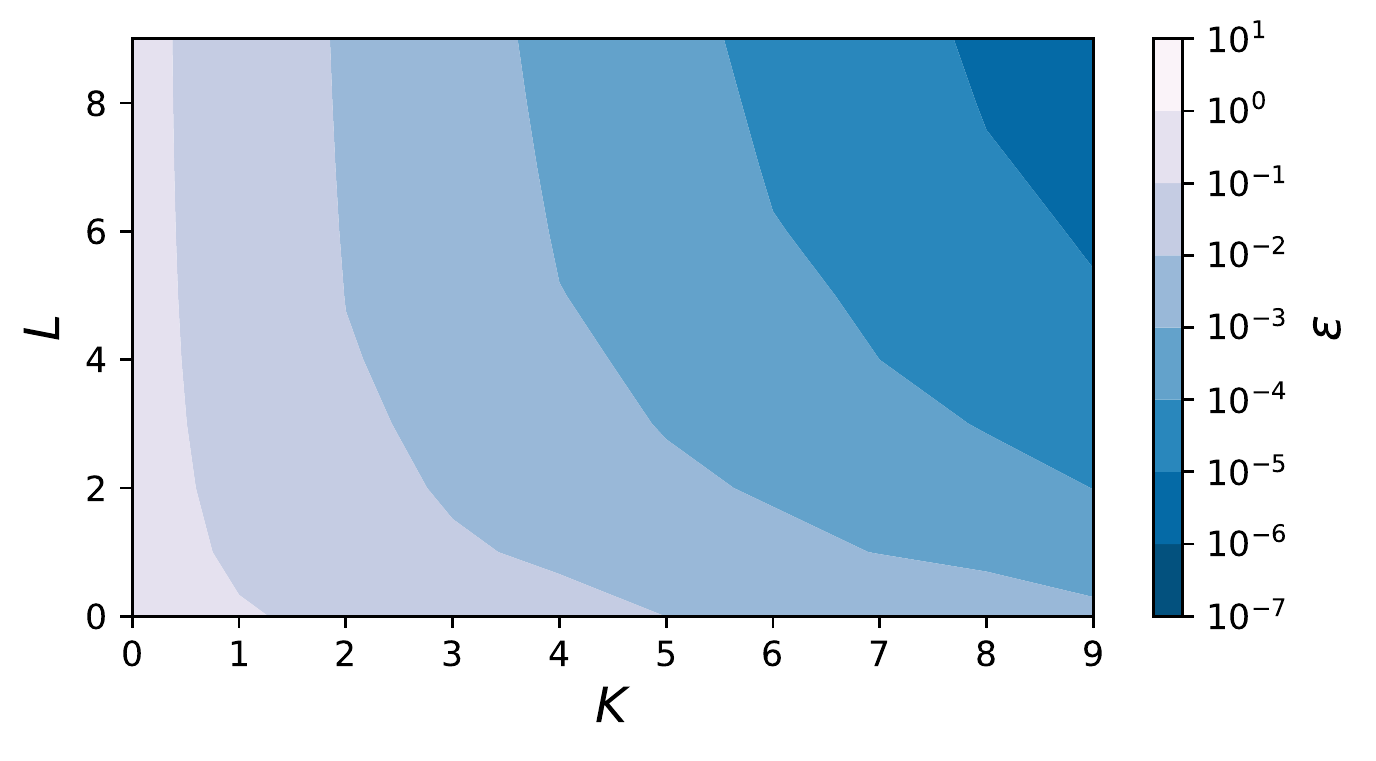}
  \caption{\emph{Illustration of the error term~$\varepsilon$ in \cref{thm:main}.} The error~$\varepsilon$ decreases exponentially for Selesnick's construction as the parameters $K$ and $L$ are increased.}\label{fig:selesnick filter error}
\end{figure}

It is instructive to consider a few features of \cref{thm:main}. Suppose $f_1[n] = \delta_{n,x}$ and $f_2[n]=\delta_{n,y}$. Then $G(\{f_i\})$ is simply the \emph{two-point function} $C(x,y)=\braket{a^\dagger_x a_y}$, which in the true ground state decays with $\lvert y-x\rvert $ as a power law, $C \sim \lvert y-x\rvert^{-1/2}$. Yet, \cref{thm:main} gives a bound that is independent of the separation $\lvert y-x\rvert$. This might seem puzzling since for a finite depth $\mathcal{L}$, all correlations between operators separated by more than $2^{\mathcal{L}} M$ vanish. However, at a distance of $\lvert y-x\rvert = M 2^{\mathcal{L}}$, the two-point function is of order $C \sim M^{-1/2} 2^{-\mathcal{L}/2}$ which is consistent with \cref{thm:main}: dropping the second term in the square root, we still have $M^{-1/2} 2^{-\mathcal{L}/2} \leq M^{1/2} 2^{-\mathcal{L}/4}$.

More generally, the two terms in the square root in \cref{thm:main} have different physical interpretations.
The first is associated with the convergence of the renormalization group transformation, while the second is associated with the goodness of approximation of the phase relation.
Indeed, \cref{eq:approxhalfdelay quantitative} requires that the phase relation~\eqref{eq:selesnickwavelet} is approximately correct or, when this is not the case for some~$k$, that both $h_s(k)$ and $g_s(k)$ are small in magnitude (cf.\ \cref{sec:1d}).

Our proof of \cref{thm:main} makes this intuition precise.
We show that \cref{eq:approxhalfdelay quantitative} guarantees that the single-particle modes prepared by the MERA are approximate eigenmodes, and the boundedness of the scaling function ensures that the truncation error decreases exponentially with the number of layers of the tensor network.
Together, this implies that the two-point correlation functions of the states $\ket\Omega$ and $\ket{\Omega_{\text{MERA}}}$ are approximately equal.
We then use a robust version of Wick's theorem~\cite{powers1970free} to show that higher correlation functions can likewise be approximated up to small error.
We refer to \cref{app:proof} for a rigorous mathematical proof.

It is remarkable that the error converges as $\mathcal L\to\infty$: even though correlation functions now depend on an infinite number of ``non-ideal" (finite $\varepsilon$) layers, the total error is bounded. This is a consequence of the hierarchical renormalization group structure of the network combined with the boundedness of the scaling functions.

Note that \cref{thm:main} does \emph{not} provide an error estimate on the fidelity between the true ground state and the MERA state for an infinite system.
Indeed, these two states are expected to necessarily be orthogonal in the thermodynamic limit, since any finite error per unit volume will result in zero overlap as the system size is taken to infinity.
Nevertheless, \cref{thm:main} proves that our construction can yield correlation functions that approximate those of the true ground state to arbitrary accuracy.
Therefore, all intensive (not scaling with system size) physical properties that can be inferred from these are likewise well-reproduced.
Our results can thus be seen as another instance where we can rigorously construct tensor network states for critical systems or for quantum field theories if we focus on correlation functions, a point first raised in~\cite{Koenig_2017,Koenig_2016}.

On a finite ring of size $V$, the one-dimensional model has an energy gap $\Delta \sim 1/V$.
In such a situation, the infinite system $\mathcal{L} \rightarrow \infty$ circuit must be modified to fit into the finite-size system.
We expect that there exists an analogue of \cref{thm:main} that guarantees correlation functions are well-approximated for sufficiently small $\varepsilon$.
Moreover, in this finite-size setting and with sufficiently small $\varepsilon$, the state $\ket{\Omega_{\text{MERA}}}$ can have high overlap with $\ket{\Omega}$. Indeed, if $P_\Omega = \ket{\Omega} \bra{\Omega}$ is the ground state projector and $E_\Omega$ is the ground state energy, then we have $ \Delta (1-P_\Omega) \leq (H - E_\Omega )$ and hence
\begin{equation*}
1 - |\langle \Omega | \Omega_{\text{MERA}} \rangle|^2 \leq \frac{1}{\Delta} (\langle \Omega_{\text{MERA}} | H | \Omega_{\text{MERA}} \rangle - E_\Omega ).
\end{equation*}
Thus if the energies of $\ket{\Omega}$ and $\ket{\Omega_{\text{MERA}}}$ are within $1/\text{poly($V$)}$ of each other, then the overlap $|\langle \Omega | \Omega_{\text{MERA}} \rangle|^2$ is within $1/\text{poly($V$)}$ of one. If, as suggested by our numerics, the error $\varepsilon$ achieved by Selesnick's wavelets say, goes down exponentially with $\min(K,L)$, then one would have high overlap between a MERA state and the true ground state using a bond dimension scaling only polynomially in~$V$.

\section{Numerical results}\label{sec:numerics}
Our construction can be used to effectively calculate physical properties in real space~\footnote{Software available at \url{https://github.com/catch22/pyfermions}.}.
For example, consider the \emph{energy density} of the approximate ground state.
Its value for the MERA quantum circuit for the infinite one-dimensional discrete line, truncated at depth~$\mathcal L$, is given by $\sum_{\ell=1}^\mathcal L 2^{-(\ell+1)} e_{(\ell)}$, with $e_{(\ell)} = -2 \Re \sum_n \phi_{(\ell)}[n] \phi^*_{(\ell)}[n+1]$ the single-particle energy of a mode~$\phi_{(\ell)}$ obtained from the $\ell$-th layer.
The scaling factor comes from the fact that at the $\ell$-th level of the MERA, the density of degrees of freedom is~$2^{-\ell}$, half of which are filled.
This can be easily be evaluated numerically and displays convergence to the true value~$-2/\pi$, as illustrated in the left panel of \cref{fig:energy}.
The numerical results are consistent with a power law convergence in the effective bond dimension $\chi=2^{K+L}$, in agreement with our discussion below \cref{thm:main}.
We find that our analytical construction systematically improves over the one from~\cite{evenbly2016entanglement} but its energy density estimate is outperformed by the variationally optimized non-Gaussian MERA from~\cite{evenbly2013quantum}.

\begin{figure}
  \includegraphics[width=0.23\textwidth]{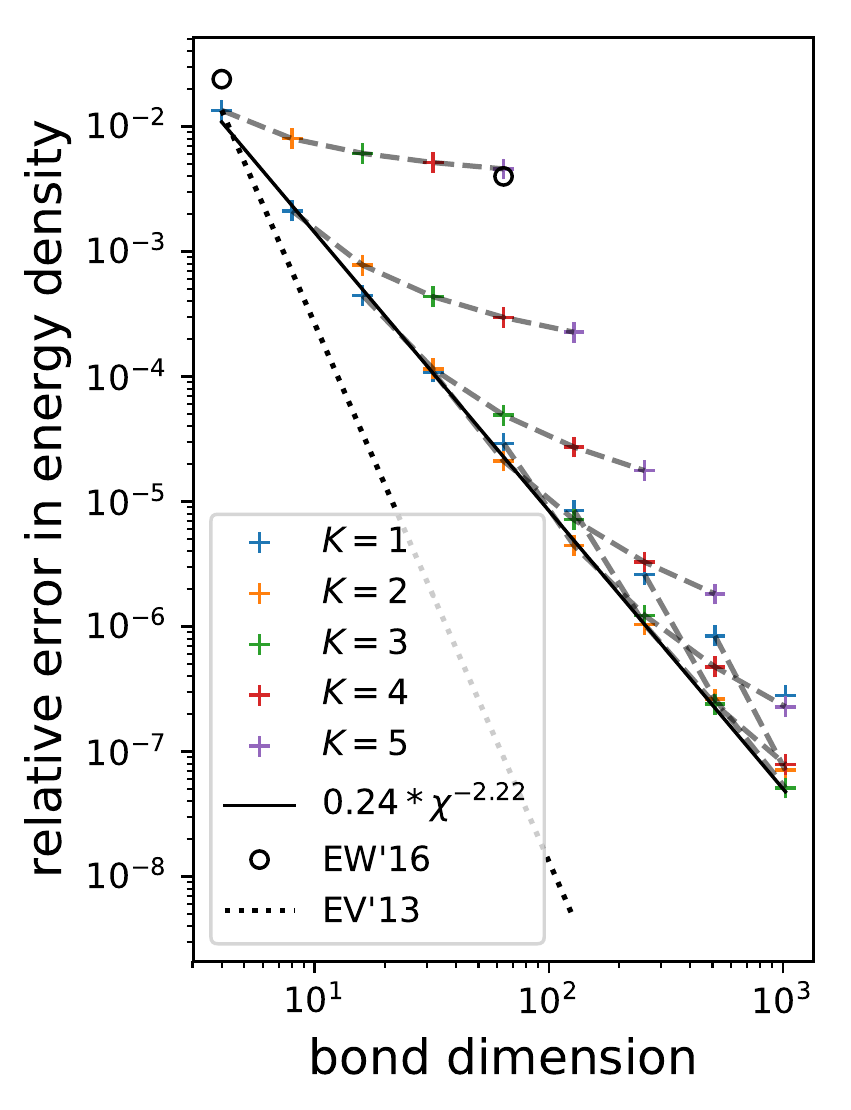}
  \includegraphics[width=0.23\textwidth]{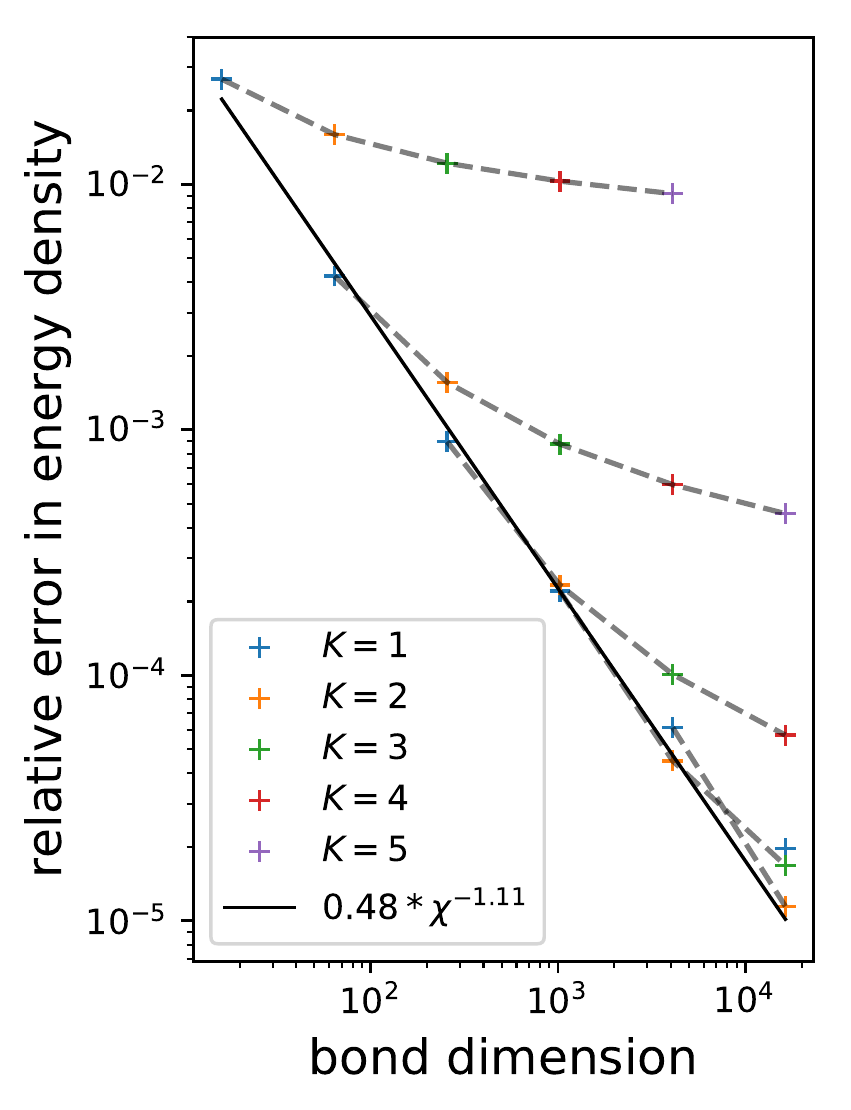}
  \caption{\emph{Relative error in the energy density} for the approximate ground states prepared by our quantum circuits in one dimension (left) and in two dimensions (right), for various values of the wavelet parameters~$K,L$ (dashed lines indicate same~$L$). For comparison we also display the relative error for the analytical construction from~\cite{evenbly2016entanglement} and for the numerically optimized non-Gaussian MERA from~\cite{evenbly2013quantum}.}\label{fig:energy}
\end{figure}

A similar procedure works for the two-di\-men\-si\-o\-nal square lattice.
The energy density is now given by $\sum_{\ell_x=1}^{\mathcal L_x} \sum_{\ell_y=1}^{\mathcal L_y} 2^{-(\ell_x+\ell_y+1)} e_{(\ell_x,\ell_y)}$, where $e_{(\ell_x,\ell_y)}$ denotes the single-particle energy of a mode obtained from levels~$\ell_x,\ell_y$ of the quantum circuit, which we recall denote the number of renormalization steps in the $x$ and $y$ direction, respectively.
In other words, $\min(\ell_x,\ell_y)$ is the level at which we switch to a one-dimensional branch (cf.~\cref{fig:mera2d}).
It is useful to note that, since the two wavelet transforms involved are separable, the modes obtained on each sublattice are tensor products of one-dimensional modes, coupled only by the final Hadamard transforms.
This allows us to carry out all computations in the one-dimensional setting.
Our numerical results are shown in the right panel of \cref{fig:energy} and again show power law convergence in the effective bond dimension to the true value~$-8/\pi^2$.

As a last example, we consider a general \emph{two-point function} $C(x,y)=\braket{a_x^\dagger a_y}$.
While the true ground state is translation-invariant, $C(x,y)\neq C(y-x)$ for the approximate ground state prepared by the quantum circuit, since the latter is not perfectly invariant under arbitrary lattice translations.
For simplicity, we only discuss the one-dimensional case.
As above, let~$\phi_{(\ell)}$ denote a single-particle mode obtained from the $\ell$-th level of the MERA quantum circuit.
Then we have $C(x,y)=\sum_{\ell=1}^\mathcal L \sum_{z\in\Z} \phi_{(\ell)}[x-2^{\ell+1} z] \phi_{(\ell)}^*[y-2^{\ell+1} z]$.
The inner sum corresponds to the different modes obtained from the $\ell$-th level, obtained as translates of $\phi_{(\ell)}$; we note that only finitely many translates yield a nonzero summand since the $\phi_{(\ell)}$ are finitely supported.
The result is shown in \cref{fig:twopoint}.
Again we find convergence to the exact solution $C(y-x)=\sin(\pi(y-x)/2)/(\pi(y-x))$.
In particular, the two-point function becomes more and more translation-invariant with increased $K,L$.

\begin{figure}
\includegraphics[width=0.23\textwidth]{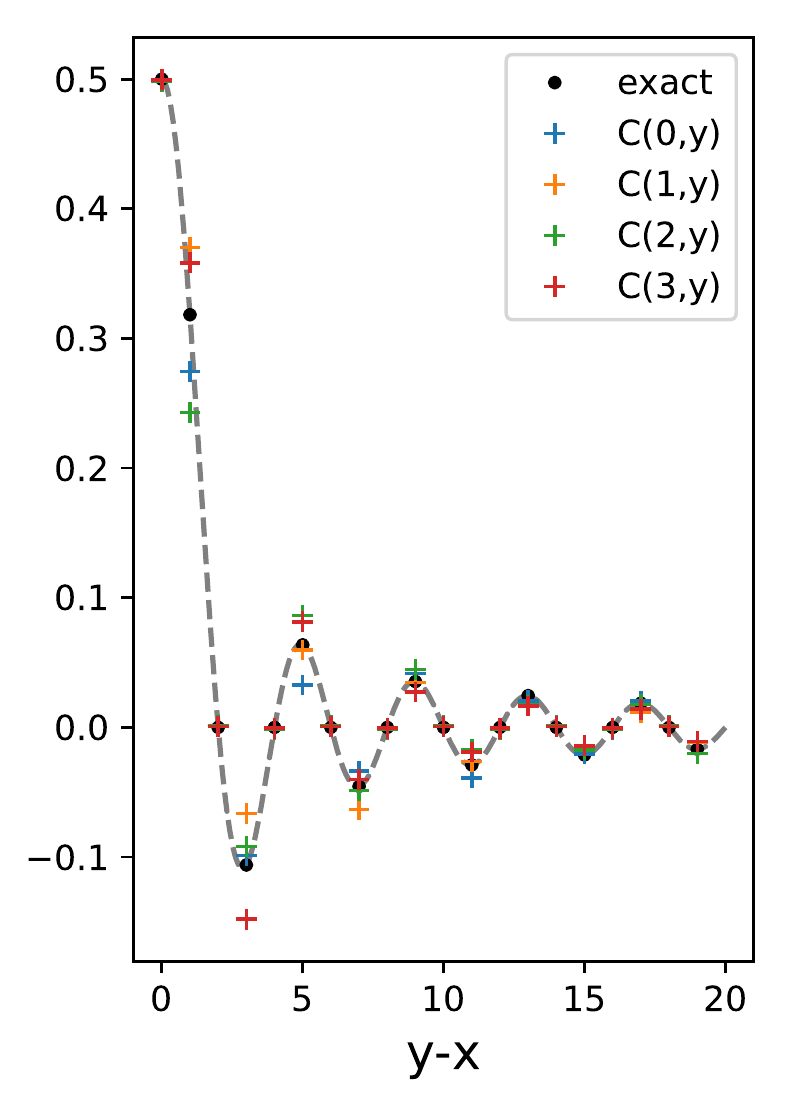}
\includegraphics[width=0.23\textwidth]{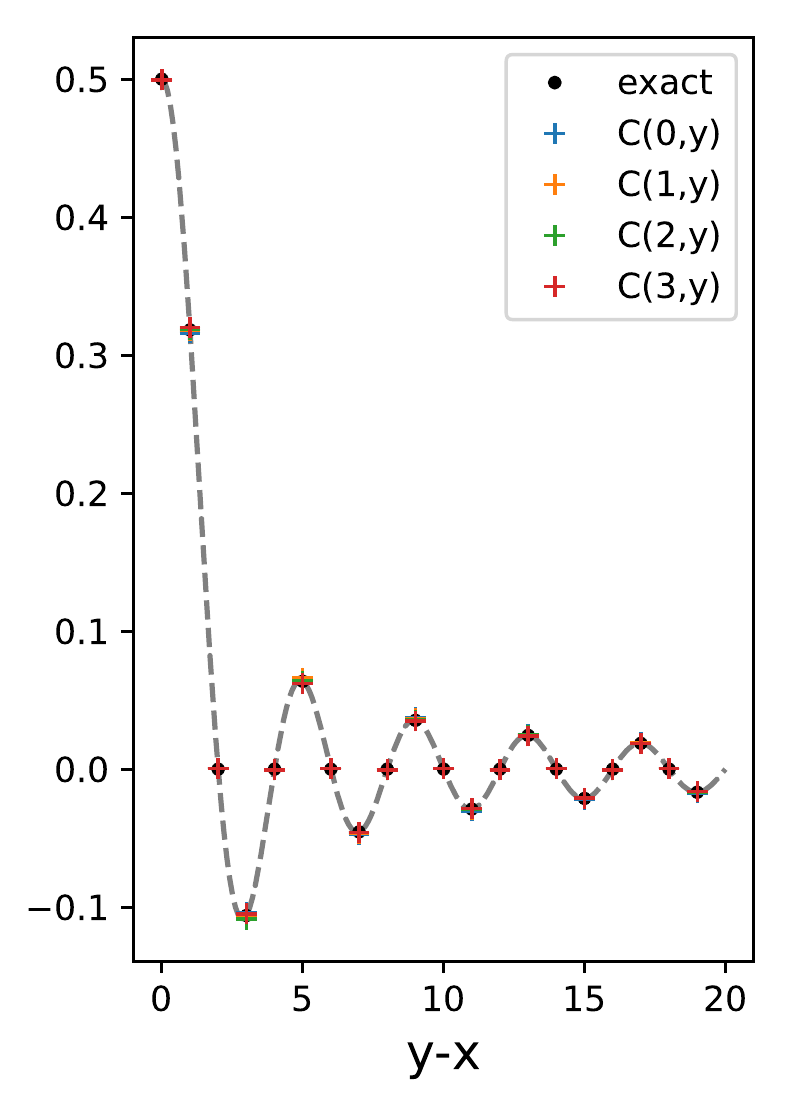}
\caption{\emph{Two-point function~$C(x,y)$} of the approximate ground states in one-dimension for wavelet parameters~$K=L=1$ (left) and $K=L=3$ (right).}\label{fig:twopoint}
\end{figure}

\section{Conclusions}\label{sec:conclusions}
In this work we showed how wavelet theory can be used to rigorously construct quantum circuits that approximate metallic states of non-interacting fermions.
Working directly in the thermodynamic limit, we showed that arbitrary correlation functions of fermion creation and annihilation operators can be approximated to high accuracy for appropriate choice of circuit parameters.
In a finite-size system, we argued based on our numerics that a tensor network with bond dimension scaling only polynomially with system size can achieve unit overlap with the true ground state in the large system size limit.
Although such a bond dimension is high from the point of view of numerical calculations using a classical computer, from an information-theoretic point of view it represents an astounding compression of the quantum state.
At no point did we use a variational optimization to determine the circuit parameters, and the circuits we construct have a plain physical meaning.
The essential physics arose from the structure of negative and positive energy eigenspaces and was encapsulated in a half-shift delay between pairs of wavelet filters.
The design of such pairs of wavelets had already been carried out in the signal processing community.

The constructions reported here are closely related to a forthcoming work by three of the authors which uses wavelet technology to approximate correlation functions in a continuum quantum field theory, namely the free Dirac field in 1+1 dimensions.
As in the case of the thermodynamic limit of the lattice system, the correct notion of approximation turns out to be approximation of correlation functions instead of approximation of states.
Using the free Dirac field construction, it is also possible to construct MERA circuits which approximate the correlation functions of interacting Wess-Zumino-Witten field theories in 1+1 dimensions.

There are many immediate directions for further development.
It is of considerable interest to adapt existing wavelets or design new wavelets that can capture curved Fermi surfaces; then we would truly be able to describe a general class of metallic states in two or more dimensions.
This would, for example, enable us to address different chemical potentials in the square lattice model.
It is also interesting to adapt our wavelet approach to describe Dirac points in two or more dimensions; the basic approach used here is clearly sound, but there is an interesting wavelet design problem to capture the physics of the filled Dirac sea.
Another very interesting direction is interacting fermions.
For example, similar in spirit to Slater-Jastrow wavefunctions, our non-interacting wavelet MERA construction might be used as the starting point for a variational class of wavefunctions for interacting metals.

\begin{acknowledgments}
We acknowledge inspiring discussions with Frank Verstraete, Karel van Acoleyen and Matthias Bal.
JH acknowledges funding from the European Commission (EC) via ERC grant ERQUAF (715861).
BGS is supported by the Simons Foundation as part of the It From Qubit Collaboration; through a Simons Investigator Award to Senthil Todadri; and by MURI grant W911NF-14-1-0003 from ARO.
MW gratefully acknowledges support from the Simons Foundation and AFOSR grant FA9550-16-1-0082.
JC is supported by the Fannie and John Hertz Foundation and the Stanford Graduate Fellowship program.
VBS was supported by the EC through grants QUTE and SIQS.
\end{acknowledgments}

\bibliography{library}

\cleardoublepage\onecolumngrid\appendix

\section{Proof of Theorem~\ref{thm:main}}\label{app:proof}

We start by describing the setup in precise mathematical terms.
Any pure gauge-invariant generalized free state $\Psi$ can be described by an operator $\psi\geq0$, known as the \emph{symbol}, such that the correlation functions are given by
\begin{equation}\label{eq:G from symbol}
  G(\{f_i\})_\Psi = \det\left[\braket{f_i | \psi | f_{N+1-j}}\right]_{i,j=1}^N.
\end{equation}
For pure states, $\psi$ is a projection which we can think of as the projection onto the Fermi sea.
To connect with physics notation note that ``gauge-invariant" means effectively that the density matrix~$\Psi$ is invariant under a global $U(1)$ (particle number) transformation of the form $a_\alpha \mapsto e^{\mathrm{i} \theta} a_\alpha$.

Both the true ground state and the state prepared by the MERA are pure gauge-invariant generalized free states.
Following the discussion in \cref{sec:1d}, their symbols can be described as follows.
We denote by~$v\colon\ell^2(\Z)\ot\C^2\to\ell^2(\Z)$ the unitary corresponding to the transformation~$(b_1,b_2) \mapsto a$ and by~$m(\theta_w)\colon\ell^2(\Z)\to\ell^2(\Z)$ the Fourier multiplier by~$\theta_w(k)=-\rmi\sign(k)e^{\rmi\frac k2}$, so that the operator~\eqref{eq:1dD} can be written as $d = \left[\begin{smallmatrix}I & 0\\0 & m(\theta_w)\end{smallmatrix}\right]$. Recall that $u=d (I \ot h_2)$, with $h_2$ the Hadamard gate.
Then the symbol of the true ground state~$\ket\Omega$ is given by
\begin{equation}\label{eq:symbol true}
\omega
= v u \begin{bmatrix}I & 0 \\ 0 & 0\end{bmatrix} u^\dagger v^\dagger
= v \begin{bmatrix}I & 0\\0 & m(\theta_w)\end{bmatrix} \left( I \ot \ket+\bra+ \right) \begin{bmatrix}I & 0\\0 & m(\theta_w)^\dagger\end{bmatrix} v^\dagger,
\end{equation}
where $\ket+ = (\ket0+\ket1)/\sqrt2$.
Next, recall that we are given two pairs of filters, $h_s,h_w$ and $g_s,g_w$.
We denote the corresponding wavelet transforms, truncated at level~$\mathcal L$, by
\[ W_h^{(\mathcal L)}, W_g^{(\mathcal L)} \colon \ell^2(\Z) \to \ell^2(\Z) \ot \C^{\mathcal L+1}, \]
where the first $\mathcal L$ coordinates of the output correspond to the wavelet outputs and the last one to the remaining scaling output (see, e.g., \cite{mallat2008wavelet} for an introduction to wavelet theory).
Let us denote by~$p_w^{(\mathcal L)},p_s^{(\mathcal L)}$ the projection onto the wavelet outputs and the scaling output, respectively.
Then the many-body ground state~$\ket\Omega_{\text{MERA}}$ prepared by the MERA quantum circuit has symbol
\begin{equation}\label{eq:symbol approx}
\omega_{\text{MERA}}
= v \begin{bmatrix}W_h^{(\mathcal L),\dagger} & 0\\0 & W_g^{(\mathcal L),\dagger}\end{bmatrix} \left( p_w^{(\mathcal L)} \ot \ket+\bra+ \right) \begin{bmatrix}W_h^{(\mathcal L)} & 0\\0 & W_g^{(\mathcal L)}\end{bmatrix} v^\dagger.
\end{equation}
Let $F \subseteq \ell^2(\Z)$ denote a subspace (which we will later take to be the span of the observables $f_i$).
Let
\[ D(F) := \sup_{f\in F} \lvert\{ x\in\Z : f(x)\neq0 \}\rvert \]
denote the maximal support of any sequence in $F$.
We denote by~$p_F$ the orthogonal projector onto $F$ and abbreviate $\omega|_F := p_F \omega p_F$.

As usual, we will write $\lVert-\rVert_p$ for $p$-norms, $\lVert-\rVert_\infty$ for supremum norms, and $\lVert-\rVert$ for operator norms.
We will use $m(\theta)\colon\ell^2(\Z)\to\ell^2(\Z)$ more generally for the Fourier multiplier by some periodic function~$\theta(k)$.

We first prove that the truncation of the MERA only incurs an error that is exponentially small in $\mathcal L$.

\begin{lem}\label{lem:scaling decay}
Let $h_s\in\ell^2(\Z)$ be a scaling filter of length $M$ such that the associated scaling function $\phi_h \in L^2(\R)$ is bounded.
Let $f \in \ell^1(\Z)$.
Then:
\[ \lVert p_s^{(\mathcal L)} W_h^{(\mathcal L)} f \rVert_2 \leq M 2^{-(\mathcal L-1)/2} \lVert \phi_h\rVert_\infty \lVert f\rVert_1 \]
\end{lem}
\begin{proof}
Let $\delta_n$ denote the sequence that is equal to one at site $n$ and zero elsewhere.
By the definition of the discrete wavelet transform, we have that
\begin{align*}
\lVert p_s^{(\mathcal L)} W_h^{(\mathcal L)} \delta_n \rVert_2^2 = \sum_{m\in\Z} \bigl\lvert\int_{-\infty}^\infty dx \, \phi_h^*(x-n) 2^{-\mathcal L/2} \phi_h(2^{-\mathcal L}x-m)\bigr\rvert^2.
\end{align*}
Since the scaling filter has length $M$, the scaling function $\phi_h$ is supported in some interval $[x_0,x_0+M-1]$, and so the above is equal to
\begin{align*}
\sum_{m\in\Z} \bigl\lvert \int_{x_0}^{x_0+M-1} dx \, \phi^*_h(x-n) 2^{-\mathcal L/2} \phi_h(2^{-\mathcal L}x-m) \bigr\rvert^2
\leq \sum_{m\in\Z} \int_{2^{-\mathcal L}x_0-m}^{2^{-\mathcal L}(x_0+M-1)-m} dy \, \lvert \phi_h(y) \rvert^2,
\end{align*}
where we have used the Cauchy-Schwarz inequality.
Since there are at most $2M$ nonzero summands, we can upper bound this by
$M^2 2^{-\mathcal L+1} \lVert \phi_h\rVert_\infty^2$.
We have thus established that
\[ \lVert p_s^{(\mathcal L)} W_h^{(\mathcal L)} \delta_n \rVert_2 \leq M 2^{-(\mathcal L-1)/2} \lVert \phi_h\rVert_\infty, \]
from which the lemma follows at once.
\end{proof}

Now recall that our two scaling filters $h_s$ and $g_s$ have length~$M$ and that the associated scaling functions are bounded in absolute value by~$B$.
Let $f\in F$.
Then $v^\dagger f = (f_h,f_g)$ where $\lVert f\rVert_1 = \lVert f_h\rVert_1 + \lVert f_g\rVert_1$, and we obtain
\begin{align*}
  \lVert (p_s^{(\mathcal L)} \ot I) \begin{bmatrix}W_h^{(\mathcal L)} & 0\\0 & W_g^{(\mathcal L)}\end{bmatrix} v^\dagger f \rVert_2
\leq \lVert p_s^{(\mathcal L)} W_h^{(\mathcal L)} f_h \rVert_2 + \lVert p_s^{(\mathcal L)} W_g^{(\mathcal L)} f_g \rVert_2
\leq M 2^{-(\mathcal L-1)/2} B \lVert f\rVert_1
\leq \sqrt{D(F)} \, B M 2^{-(\mathcal{L}-1)/2} \lVert f\rVert_2,
\end{align*}
where the second inequality is \cref{lem:scaling decay} applied to both $h_s$ and $g_s$; the last inequality is the Cauchy-Schwarz inequality. Therefore:
\begin{equation}\label{eq:error truncation 1}
  \lVert p_F v \begin{bmatrix}W_h^{(\mathcal L),\dagger} & 0\\0 & W_g^{(\mathcal L),\dagger}\end{bmatrix} (p_s^{(\mathcal L)} \ot I) \rVert
\leq \sqrt{D(F)} \, B M 2^{-(\mathcal L-1)/2}
\end{equation}
The same argument establishes that
\begin{equation}\label{eq:error truncation 2}
  \lVert p_F v \begin{bmatrix}I & 0\\0 & m(\theta_w)\end{bmatrix} (W_h^{(\mathcal L),\dagger} p_s^{(\mathcal L)} \ot I) \rVert
\leq \sqrt{D(F)} \, B M 2^{-(\mathcal L-1)/2}.
\end{equation}
We now show that the MERA generates approximate eigenmodes.

\begin{lem}\label{lem:phase error}
Let $h_s, g_s$ be scaling filters such that \cref{eq:approxhalfdelay quantitative} holds.
Then we have for all $\ell=1,\dots,\mathcal L$ that
\[ \lVert W_g^{(\mathcal L),\dagger} (I \ot \ket \ell) - m(\theta_w) W_h^{(\mathcal L),\dagger} (I \ot \ket \ell) \rVert \leq \varepsilon \ell. \]
\end{lem}
\begin{proof}
  We start with the formula
  \begin{equation}\label{eq:inverse wavelet}
    W_h^{(\mathcal L),\dagger} (I \ot \ket \ell) = \bigl[ m(h_s) \uparrow \bigr]^{\ell-1} m(h_w) \uparrow = \underbrace{m(h_s) \uparrow \ldots m(h_s) \uparrow}_{\ell-1 \text{ times}} m(h_w) \uparrow,
  \end{equation}
  where $\uparrow$ denotes the upsampling operator on $\ell^2(\Z)$, defined by~$\delta_n \mapsto \delta_{2n}$.

  Now recall that $\theta_w(k) = -\rmi\sign(k)e^{\rmi\frac k2}$.
  Let us define $\theta_s(k) = e^{-\rmi\frac k2}$.
  It is easily verified that $\theta_s(k) \theta_w(2k) = \theta_w(k)$,
  which can equivalently be written as $m(\theta_s) \uparrow m(\theta_w) = m(\theta_w)$.
  Using \cref{eq:inverse wavelet} and iteratively applying this relation,
  \begin{align*}
    \bigl[W_g^{(\mathcal L),\dagger} - m(\theta_w) W_h^{(\mathcal L),\dagger}\bigr] (I \ot \ket \ell)
  = \bigl[m(g_s) \uparrow\bigr]^{\ell-1} m(g_w) \uparrow - \bigl[m(\theta_s h_s) \uparrow\bigr]^{\ell-1} m(\theta_w h_w) \uparrow,
  \end{align*}
  which can be written as a telescoping sum,
  \begin{align*}
  \left( \sum_{j=0}^{\ell-2} \bigl[m(\theta_s h_s) \uparrow\bigr]^j \bigl[m(g_s) - m(\theta_s h_s)\bigr] \uparrow \bigl[m(g_s) \uparrow\bigr]^{\ell-2-j} m(g_w) \uparrow \right) + \bigl[m(\theta_s h_s) \uparrow\bigr]^{\ell-1} \bigl[m(g_w) - m(\theta_w h_w)\bigr] \uparrow.
  \end{align*}
  The unitary of the wavelet transform implies that all four maps $m(h_s) \uparrow$, $m(h_w) \uparrow$, $m(g_s) \uparrow$, $m(g_w) \uparrow$ are isometries.
  Since the upsampling operator $\uparrow$ is an isometry and the Fourier multipliers $m(\theta_s)$, $m(\theta_w)$ are clearly unitaries, we obtain the desired bound
  \begin{align*}
    \lVert [W_g^{(\mathcal L),\dagger} - m(\theta_w) W_h^{(\mathcal L),\dagger}] (I \ot \ket \ell) \rVert
  \leq (\ell-1) \lVert m(g_s) - m(\theta_s h_s) \rVert + \lVert m(g_w) - m(\theta_w h_w) \rVert
  \leq \varepsilon \ell.
  \end{align*}
  For the second inequality, we note that \cref{eq:approxhalfdelay quantitative} is not only equivalent to $\lVert m(g_s) - m(\theta_s h_s) \rVert \leq\varepsilon$, but it also ensures that $\lVert m(g_w) - m(\theta_w h_w) \rVert\leq\varepsilon$, which follows from the relation $h_w(k) = e^{ik} h_s^*(k+\pi)$ and its analogue for $g_w,g_s$.
\end{proof}

It follows directly from \cref{lem:phase error} that
\begin{equation}\label{eq:phase error}
  \lVert W_g^{(\mathcal L),\dagger} p_w^{(\mathcal L)} - m(\theta_w) W_h^{(\mathcal L),\dagger} p_w^{(\mathcal L)} \rVert
  \leq \varepsilon \sum_{\ell=1}^\mathcal L \ell
  \leq \varepsilon \mathcal L^2.
\end{equation}
However, this upper bound can be arbitrarily large.
We will show how to circumvent this issue.

\begin{lem}\label{lem:tradeoff}
Under the assumptions of \cref{thm:main}, we have that
\[ \lVert p_F v \begin{bmatrix}W_h^{(S),\dagger} & 0\\0 & W_g^{(\mathcal L),\dagger}\end{bmatrix} (p_w^{(\mathcal L)} \ot I) - p_F v \begin{bmatrix}I & 0\\0 & m(\theta_w)\end{bmatrix} (W_h^{(\mathcal L),\dagger} \ot I) \rVert \leq C 2^{-\mathcal L/2} + 6 \varepsilon \log_2^2(C/\varepsilon), \]
where $C := 2^{3/2} \sqrt{D(F)} \, B M \geq 2$.
\end{lem}
\begin{proof}
  Let $\mathcal L'\in\{1,\dots,\mathcal L\}$ and write $q^{(\mathcal L',\mathcal L)}\colon \ell^2(\Z) \ot \C^{\mathcal L+1} \to \ell^2(\Z) \ot \C^{\mathcal L'+1}$ for the projection onto the first $\mathcal L'+1$ components.
  It follows from the hierarchical form of the wavelet transform that $W_h^{(\mathcal L),\dagger}$ and $W_h^{(\mathcal L'),\dagger} p_w^{(\mathcal L')} q^{(\mathcal L',\mathcal L)}$ and differ by a term that is the composition of $W_h^{(\mathcal L'),\dagger} p_s^{(\mathcal L')}$ with a partial isometry; likewise for the other wavelet transform.
  Thus, \cref{eq:error truncation 2} implies that
  \begin{align*}
  \lVert p_F v \begin{bmatrix}I & 0\\0 & m(\theta_w)\end{bmatrix} (W_h^{(\mathcal L),\dagger} \ot I) - p_F v \begin{bmatrix}I & 0\\0 & m(\theta_w)\end{bmatrix} (W_h^{(\mathcal L'),\dagger} p_w^{(\mathcal L')} q^{(\mathcal L',\mathcal L)} \ot I) \rVert
  \leq \sqrt{D(F)} \, B M 2^{-(\mathcal L'-1)/2}.
  \end{align*}
  Similarly, \cref{eq:error truncation 1} together with the observation that $p_w^{(\mathcal L')} q^{(\mathcal L',\mathcal L)} p_w^{(\mathcal L)} = p_w^{(\mathcal L')} q^{(\mathcal L',\mathcal L)}$ implies that
  \begin{align*}
  \lVert p_F v \begin{bmatrix}W_h^{(\mathcal L),\dagger} & 0\\0 & W_g^{(\mathcal L),\dagger}\end{bmatrix} (p_w^{(\mathcal L)} \ot I) - p_F v \begin{bmatrix}W_h^{(\mathcal L'),\dagger} & 0\\0 & W_g^{(\mathcal L'),\dagger}\end{bmatrix} (p_w^{(\mathcal L')} q^{(\mathcal L',\mathcal L)} \ot I) \rVert
  \leq \sqrt{D(F)} \, B M 2^{-(\mathcal L'-1)/2}.
  \end{align*}
  On the other hand, \cref{eq:phase error} ensures that
  \begin{align*}
  \lVert p_F v \begin{bmatrix}I & 0\\0 & m(\theta_w)\end{bmatrix} (W_h^{(\mathcal L'),\dagger} p_w^{(\mathcal L')} q^{(\mathcal L',\mathcal L)} \ot I) - p_F v \begin{bmatrix}W_h^{(\mathcal L'),\dagger} & 0\\0 & W_g^{(\mathcal L'),\dagger}\end{bmatrix} (p_w^{(\mathcal L')} q^{(\mathcal L',\mathcal L)} \ot I) \rVert
  \leq \varepsilon \mathcal L'^2.
  \end{align*}
  By combining the above bounds we obtain that
  \[ \lVert p_F v \begin{bmatrix}W_h^{(\mathcal L),\dagger} & 0\\0 & W_g^{(\mathcal L),\dagger}\end{bmatrix} (p_w^{(\mathcal L)} \ot I) - p_F v \begin{bmatrix}I & 0\\0 & m(\theta_w)\end{bmatrix} (W_h^{(\mathcal L),\dagger} \ot I) \rVert
  \leq \sqrt{D(F)} \, B M 2^{-(\mathcal L'-3)/2} + \varepsilon \mathcal L'^2 = C 2^{-\mathcal L'/2} + \varepsilon \mathcal L'^2. \]
  We can still optimize this expression over $\mathcal L'\in\{1,\dots,\mathcal L\}$.
  For this, we distinguish two cases:
  If $2^{-\mathcal L/2}>\varepsilon/C$ then we choose $\mathcal L'=\mathcal L$, leading to the bound $C 2^{-\mathcal L/2} + 4 \varepsilon \log_2^2(C/\varepsilon)$.
  Otherwise, if $2^{-\mathcal L/2}\leq\varepsilon/C$, we can choose $\mathcal L'=\lfloor 2\log_2(C/\varepsilon)\rfloor$ and obtain the bound $2 \varepsilon + 4 \varepsilon \log_2^2(C/\varepsilon)$.
  In either case it is true that
  \[ \min_{\mathcal L'} \left(C 2^{-\mathcal L'/2} + \varepsilon \mathcal L'^2\right)
  \leq \max \{ C 2^{-\mathcal L/2}, 2 \varepsilon \} + 4 \varepsilon \log_2^2(C/\varepsilon)
  \leq C 2^{-\mathcal L/2} + 6 \varepsilon \log_2^2(C/\varepsilon). \qedhere \]
\end{proof}

We can at last establish our main result.

\begin{proof}[Proof of \cref{thm:main}]
We choose $F$ as the span of the observables $f_1,\dots,f_{2N}$, so that $D(F)=D(\{f_i\})$ and $\dim F\leq 2N$.
Let $\delta := C 2^{-\mathcal L/2} + 6 \varepsilon \log_2^2(C/\varepsilon)$.
We will first establish that $\lVert \omega|_F - \omega_{\text{MERA}}|_F \rVert \leq 2 \delta$.
For this, we note that
\begin{align*}
\omega|_F - \omega_{\text{MERA}}|_F
= p_F v \begin{bmatrix}I & 0\\0 & m\end{bmatrix} \left( I \ot \ket+\bra+ \right) \begin{bmatrix}I & 0\\0 & m^\dagger\end{bmatrix} v^\dagger p_F
- p_F v \begin{bmatrix}W_h^{(\mathcal L),\dagger} & 0\\0 & W_g^{(\mathcal L),\dagger}\end{bmatrix} \left( p_w^{(\mathcal L)} \ot \ket+\bra+ \right) \begin{bmatrix}W_h^{(\mathcal L)} & 0\\0 & W_g^{(\mathcal L)}\end{bmatrix} v^\dagger p_F
\end{align*}
where we have inserted \cref{eq:symbol true,eq:symbol approx}.
We now use the triangle inequality and \cref{lem:tradeoff} twice to obtain
\begin{equation}\label{eq:single-particle bound}
\begin{aligned}
\lVert \omega|_F - \omega_{\text{MERA}}|_F \rVert
\leq &\bigl\lVert \left( W_h^{(\mathcal L)} \ot I \right) \begin{bmatrix}I & 0\\0 & m^\dagger\end{bmatrix} v^\dagger p_F - \left( p_w^{(\mathcal L)} \ot I \right) \begin{bmatrix}W_h^{(\mathcal L)} & 0\\0 & W_g^{(\mathcal L)}\end{bmatrix} v^\dagger p_F \bigr\rVert \\
+ &\bigl\lVert p_F v \begin{bmatrix}I & 0\\0 & m\end{bmatrix} \left( W_h^{(\mathcal L),\dagger} \ot I \right) - p_F v \begin{bmatrix}W_h^{(\mathcal L),\dagger} & 0\\0 & W_g^{(\mathcal L),\dagger}\end{bmatrix} \left( p_w^{(\mathcal L)} \ot I \right) \bigr\rVert \leq 2 \delta.
\end{aligned}
\end{equation}
We now show that
\begin{equation}\label{eq:many-body bound}
  \lvert G(\{f_i\})_\Omega - G(\{f_i\})_{\Omega_{\text{MERA}}}\rvert \leq 24 \lVert f_1 \rVert \cdots \lVert f_{2N} \rVert \sqrt{N \delta}.
\end{equation}
For this, let us denote by~$\Omega|_F$ and $\Omega_{\text{MERA}}|_F$ the mixed gauge-invariant generalized free states with symbols~$\omega|_F$ and $\omega_{\text{MERA}}|_F$, respectively, which capture all $n$-point functions with observables from $F$.
It is clear from \cref{eq:G from symbol} that
\[ \lvert G(\{f_i\})_\Omega - G(\{f_i\})_{\Omega_{\text{MERA}}}\rvert
= \lvert G(\{f_i\})_{\Omega|_F} - G(\{f_i\})_{\Omega_{\text{MERA}}|_F}\rvert
\leq \lVert f_1 \rVert \cdots \lVert f_{2N} \rVert \cdot \lVert \Omega|_F - \Omega_{\text{MERA}}|_F \rVert_1, \]
where $\lVert-\rVert_1$ denotes the trace norm.
We now use a result by Powers and St\o{}rmer to bound the trace norm distance between generalized free states in terms of the trace norm distance of their symbol.
Specifically, we use~\cite[Lemmas~4.1 and 4.6]{powers1970free} to obtain the first inequality in
\[ \lVert \Omega|_F - \Omega_{\text{MERA}}|_F \rVert_1
\leq 12 \sqrt{\lVert \omega|_F - \omega_{\text{MERA}}|_F \rVert_1}
\leq 12 \sqrt{2N \lVert \omega|_F - \omega_{\text{MERA}}|_F \rVert}
\leq 24 \sqrt{N \delta} \]
(as long as the right-hand side is smaller than~$1/6$); for the second inequality we used that $\dim F\leq2N$ and the last one is \cref{eq:single-particle bound}.
We have thus established \cref{eq:many-body bound}, and thereby the theorem.
\end{proof}

\end{document}